\documentclass[a4paper,10pt,twocolumn]{article}
\usepackage[pdftex,bookmarksnumbered=true]{hyperref}

\usepackage{vmargin}%
\setpapersize{A4} %
\setmargrb{0.74in}{0in}{0.74in}{1in}

\usepackage[pdftex]{graphicx}

\usepackage[ruled, vlined]{algorithm2e}
\usepackage{verbatim}
\usepackage{amsmath, amsthm, amssymb}
\usepackage{pgf}

\newtheorem{theorem}{Theorem}
\newtheorem{lemma}{Lemma}
\newtheorem{definition}{Definition}

\newcommand{\adds}{\stackrel{\text{\tiny{+}}}{\gets}}

\usepackage{float}
\usepackage{amsmath,amsfonts}
\usepackage{amssymb}
\usepackage{amstext}
\usepackage[american]{babel}
\usepackage{xspace}
\usepackage{graphicx}
\usepackage{epsfig}
\usepackage{verbatim}
 \floatstyle{ruled}
 \newfloat{algorithm}{H}{alg}
 \floatname{algorithm}{Algorithm}
 \newfloat{figure}{H}{fig}
 \floatname{figure}{Figure}

\renewcommand{\familydefault}{ppl}
\usepackage[font={sf,bf},textfont=md]{caption}

\author{Vandy {\scshape Berten}, Chi-Ju {\scshape Chang}, Tei-Wei {\scshape Kuo}
 ~ \\
 ~ \\
National Taiwan University\\
Computer Science and Information Engineering dept.\\
\{vberten, ktw\}@csie.ntu.edu.tw, james299kimo@gmail.com \\
}

\title{Discrete Frequency Selection of Frame-Based Stochastic Real-Time Tasks}

\begin{document}

\maketitle
\thispagestyle{empty}
\begin{abstract}
\begin{small}
Energy-efficient real-time task scheduling has been actively explored in the past decade. Different from the past work, this paper considers schedulability conditions for stochastic real-time tasks. A schedulability condition is first presented for frame-based stochastic real-time tasks, and several algorithms are also examined to check the schedulability of a given strategy. An approach is then proposed based on the schedulability condition to adapt a continuous-speed-based method to a discrete-speed system. The approach is able to stay as close as possible to the continuous-speed-based method, but still guaranteeing the schedulability. It is shown by simulations that the energy saving can be more than 20\% for some system configurations.
\end{small}
\end{abstract}

\section{Introduction}
In the past decade, energy efficiency has received a lot of attention in system designs, ranged from server farms to embedded devices. With limited energy supply but an increasing demand on system performance, how to deal with energy-efficient real-time task scheduling in embedded systems has become a highly critical issue. There are two major ways in frequency changes of task executions: Inter-task or intra-task dynamic voltage scaling (DVS). Although Intra-task DVS seems to save more energy, the implementation is far more complicated than Inter-task DVS. Most of the time we need very good supports from compilers or/and operating systems, that is often hard to receive for many embedded systems. On the other hand, inter-task DVS is easier to deploy, and tasks might not be even aware of the deployment of the technology.

Energy-efficient real-time task scheduling has been actively explored in the past decade. Low-power real-time systems with stochastic or unknown duration have been studied for several years. The problem has first been considered in systems with only one task, or systems in which each task gets a fixed amount of time. Gruian \cite{Gruian01,Gruian01b} or Lorch and Smith \cite{pace01,pace04} both shown that when intra-task frequency change is available, the more efficient way to save energy is to increase progressively the speed. 
Solutions using a discrete set of frequencies and taking speed change overhead into account have also been proposed \cite{Xu04,Xu07b}. For inter-task frequency changes, some work has been already undertaken. In \cite{Mosse00}, authors consider a similar model to the one we consider here, even if this model is presented differently. 
The authors present several dynamic power management techniques: Proportional, Greedy or Statistical. They don't really take the distribution of number of cycles into account, but only its maximum, and its average for Statistical. According to the strategy, a task will give its slack time (the difference between the worst case and the actual number of used cycle) either to the next task in the frame, or to all of them.
In \cite{Aydin01}, authors attempt to allow the manager to tune this aggressiveness level, while in \cite{Xu07b}, they propose to adapt automatically this aggressiveness using the distribution of the number of cycles for each task. 
The same authors have also proposed a strategy taking the number of available speeds into account from the beginning, instead of patching algorithms developed for continuous speed processors \cite{Xu07}. Some multiprocessor extensions have been considered in \cite{Chen07}.

Although excellent research results have been proposed for energy-efficient real-time task scheduling, little work is done for stochastic real-time tasks, where the execution cycles of tasks might not be known in advance. In this paper, we are interested in frame-based stochastic real-time systems with inter-task DVS, where frame-based real-time tasks have the same deadline (also referred as the frame). Note that the frame-based real-time task model does exist in many existing embedded system designs, and the results of this paper can provide insight in the designs of more complicated systems. Our contribution is twofold: First, we propose a schedulability test, allowing to easily know if a frequency selection will allow to meet deadlines for any task in the system. 
As a second contribution, we provide a general method allowing to adapt a method designed for a continuous set of speeds (or frequencies) into a discrete set of speeds. This can be done more efficiently than classically by using the schedulability condition we give in the first part.
Apart from this alternative way of adapting continuous strategy, we will show how this schedulability test can be used in order to improve the robustness to parameters variation. The capability of the proposed approach is demonstrated by a set of simulations, and we show that the energy saving can be more than 20\% for some system configurations.

The rest of this paper is organized as follows: we first present the mathematical model of a real-time system that we consider in Section~\ref{sec:model}. We then present our first contribution in Section~\ref{sec:sched}, which consists in schedulability conditions and tests for the model. 
We then use those results in Section~\ref{sec:discret} and \ref{sec:ext} to explain how we can improve the discretization of continuous-speed-based strategies, and show the efficiency of this approach in the experimental part, in Section~\ref{sec:experim}, and finally conclude in Section~\ref{sec:ccl}.

\section{Model}
\label{sec:model}
We have $N$ tasks $\{T_i, i\in [1,\dots,N]\}$ which run on a DVS CPU. They all share the same deadline and period $D$ (which we call the \emph{frame}), and are executed in the order $T_1$, $T_2$, \dots, $T_N$. The maximum execution number of cycles of $T_i$ is $w_i$. Task $T_i$ will require $x$ cycles with a probability $c_i(x)$, where $c_i(\cdot)$ is then the distribution of the number of cycles. Of course, in practical, we cannot use a so precise information, and authors usually group cycles in ``bins''. For instance, we can choose to use a fixed bin system, with $b_i$ the size of the bins. In this case, the probability distribution $c'_i(\cdot)$ is such that $c'_i(k)$ represent the probability to use between $(k-1)\times b_i$ (excluded) and $k \times b_i$ (included) cycles.

The system is said to be \emph{expedient} if a task never waits intentionally. In other words, $T_1$ starts at time $0$, $T_2$ starts as soon as $T_1$ finishes, and so on.

The CPU can run at $M$ frequencies (or speeds) $f_1<f_2<\dots<f_M$, and the chosen frequency does not change during task execution. The mode $j$ consumes $P_j$ Watts.

We assume we have $N$ \emph{scheduling functions} $S_i(t)$ for $i \in [1,\dots, N]$ and $t \in [0, D]$. This function means that if $T_i$ starts its execution at time $t$, it will run until its end at frequency $S_i(t)$, where $S_i(t)\in \{f_1, f_2, ..., f_M\}$. $S_i(t)$ is then a step function (piece-wise constant function), with only $M$ possible values.
Remark that $S_i(t)$ is not necessarily an increasing or a monotonous function. This model generalizes several scheduling strategies proposed in the literature, such as  \cite{Xu07, Xu07b} -- where they consider a function corresponding to $S_i(D-t)$ --, or discrete versions of \cite{Mosse00}.
Figure~\ref{Sit} shows an example of such scheduling function set.

A scheduling function can be represented by a set of points (black dots on Figure~\ref{Sit}), representing the beginning of the step. $\mid S_i \mid$ is the number of steps of $S_i$. $S_i[k], k\in\{1, \dots, \mid S_i \mid\}$ is one point, with $S_i[k].t$ being its time component, and $S_i[k].f$ the frequency. $S_i$ has then the same value $S_i[k].f$ in the interval $\Big[S_i[k].t, S_i[k+1].t\Big[$ (with $S_i[\mid{}S_i\mid+1].t = \infty$), and we have $S_i(t) = S_i[k].f$, where
$$
k = \max \Big\{j\in \{1, \dots, \mid S_i \mid \} :  S_i[j].t \le t\Big\}.
$$
Notice that finding $k$ can be done in $\mathcal{O}(\log\mid{}S_i\mid)$ (by binary search), and, except in the case of very particular models, $\mid S_i \mid \le M$.

We first assume that changing CPU frequency does not cost any time or energy. See Section~\ref{sec:overhead} for extensions.

The scheduling functions $S_i(t)$ can be pretty general, but have to respect some constraints in order to ensure the system schedulability and avoid deadline misses.


\begin{figure}[!ht]
\begin{center}

\scalebox{0.52}{
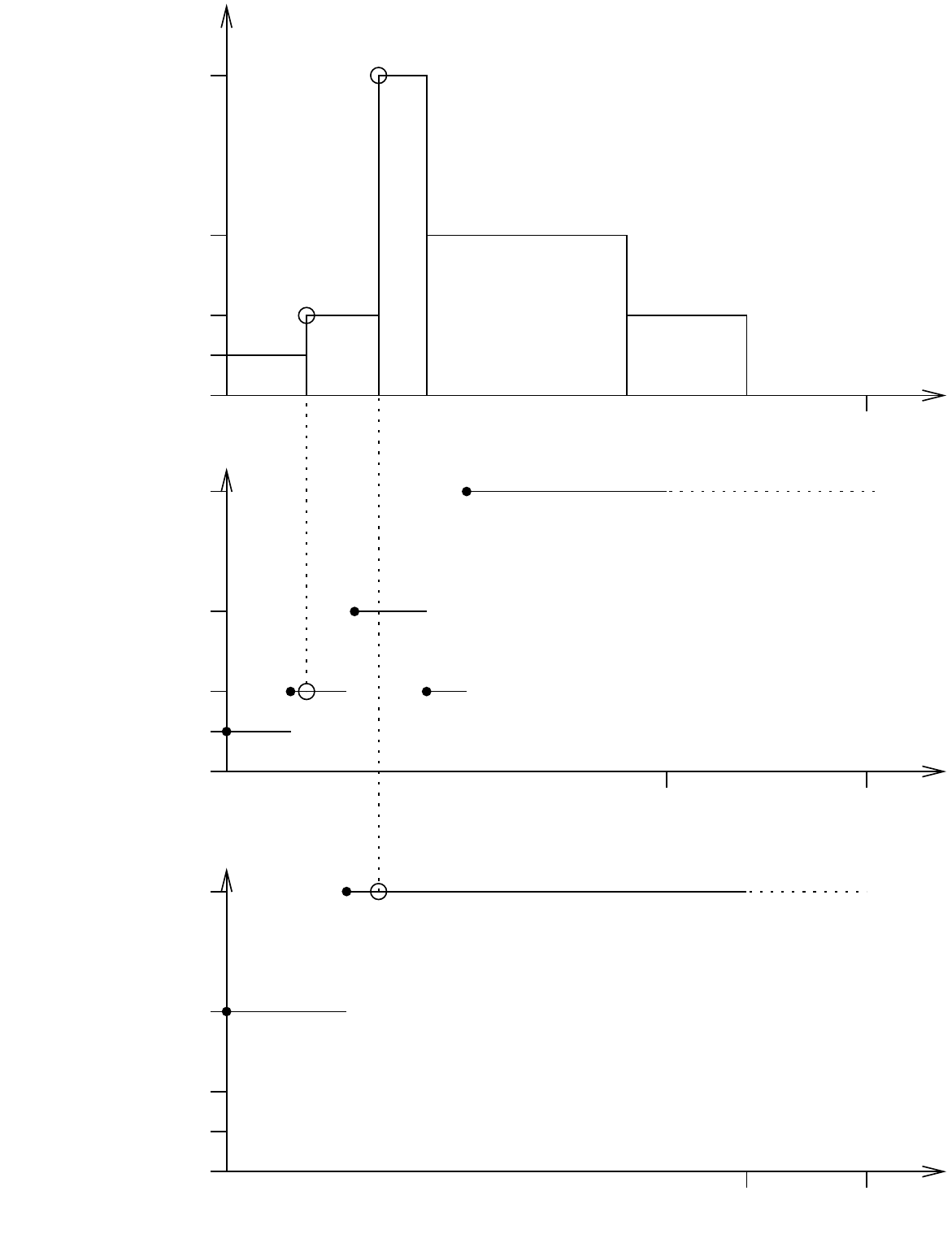
}
\caption{\label{Sit} Example of scheduling with function $S_i(t)$. We have 5 tasks $T_1, \dots, T_5$, running every $D$. In this frame, $T_1$ is run at frequency $f_1 = S_1(t_1)$, $T_2$ at $f_2 = S_2(t_2)$, $T_3$ at $f_4 = S_3(t_3)$, etc}
\end{center}
\end{figure}

We need now to define the concept of schedulability in our model:
\begin{definition}
An expedient system $\{T_i, S_i(\cdot)\}, \{f_j\} (i \in \{1, \dots, N\}, j\in \{1, \dots, M\})$ is said to be \textbf{schedulable} if, whatever the combination of effective number of cycles for each task, any task $T_i$ finishes its execution no later than the end of the frame.
\end{definition}

From this definition, we can easily see that if $\{T_i\}$ is such that $\frac{1}{f_M} \sum_{i=1}^N w_i >D$ (the left hand size represents the time needed to run any task in the frame at the highest speed if every task requires its worst case execution cycle), the system will never be schedulable, whatever the set of scheduling functions. In the same way, we can see that if $\{T_i\}$ is such that $\frac{1}{f_1} \sum_{i=1}^N w_i \le D$, the system is always schedulable, even with a ``very bad'' set of scheduling functions.

Of course, a non schedulable system could be able to run its tasks completely in almost every case. Being non schedulable means that stochastically certainly (with a probability equal to 1), we will have a frame where a task will not have the time to finish before the deadline (or the end of the frame)

\section{Schedulability and Discretization}
\label{sec:sched}

\subsection{Danger Zone}

\begin{lemma}
   Any task in $\{T_i, T_{i+1}, \dots, T_N\}$ can always finish no later than $D$ if and only if the system is expedient, and $T_i$ starts no later than $z_i$, defined as
$$
z_i = D-\frac{1}{f_M} \sum_{k=i}^N w_k.
$$
\end{lemma}

\begin{proof}
This lemma can be proved by induction.
~\\
\textbf{Initialization.}~~We first consider the case $T_N$. The very last time the task $T_N$ can start is the time allowing it to end before $D$ even if it consumes its $w_N$ cycles. At the highest frequency $f_M$, $T_N$ takes at most $\dfrac{w_N}{f_M}$ to finish. $T_N$ has then necessarily to start no later than $D-\dfrac{w_N}{f_M}$. Otherwise, if the task starts after that time, even at the highest frequency, there is no certitude that $T_N$ will finish by $D$.

~\\
\textbf{Induction.}~~We know that if (and only if) $T_{i+1}$ starts no later than $z_{i+1}$, the schedulability of $\{T_{i+1}, \dots, T_N\}$ is ensured. We need then to show that if $T_i$ starts no later than $z_i$, it will be finished by $z_{i+1}$. If $T_i$ starts no later that $z_i$, we can choose the frequency in order that $T_i$ finishes before
$$
z_i + \frac{w_i}{f_M} = D-\frac{1}{f_M} \sum_{k=i}^N w_k + \frac{w_i}{f_M} = z_{i+1}.
$$
\end{proof}

\begin{definition}
The \textbf{danger zone} of $T_i$ is the range $] z_i, D]$.
\end{definition}
This danger zone means that if $T_i$ has to start in $]z_i, D]$, we cannot guarantee the schedulability anymore. Even if, because of the variable nature of execution time, we cannot guarantee that some task will miss its deadline.
Of course, the size of the danger zone of $T_i$ is larger that the one of $T_j$ if $i<j$, which means that $z_i<z_j$ iff $i<j$.

In order to simplify some notation, we will state $z_{N+1}=D$.

\subsection{Schedulability Conditions}
Let us now consider conditions on $\{S_i\}$ allowing to guarantee the schedulability of the system. We prove the following theorem:

\begin{theorem}
 $$S_i(t) \ge \dfrac{w_i}{z_{i+1}-t} ~\forall i \in [1, \dots, N], t \in [0, z_i[,$$
 where
 $$z_i = D-\frac{1}{f_M}\sum_{k=i}^N w_k,$$
 is a necessary and sufficient condition in order to guarantee that if task $T_i$ does never require more than $w_i$ cycles and the system is expedient, any task $T_i$ can finish no later than $z_{i+1}$, and then the last one $T_N$ no later than $D$.
\end{theorem}

\begin{proof}
We show this by induction. Let $\tau_i$ be the worst finishing time of task $T_i$. Please note that this does not necessarily correspond to the case where any task before $T_i$ consumes its WCEC. Figure \ref{WC} highlights why.

\begin{figure}[!ht]
\begin{center}
\scalebox{0.80}{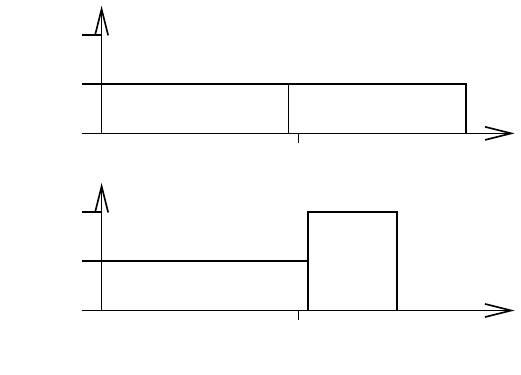}
\caption{\label{WC} Example showing that a shorter number of cycles for one task can result in a worse ending time for subsequent tasks. Here, $t'$ is the point at which $S_2(t)$ goes from $f_1$ to $f_2$. On the top plot, $T_1$ uses slightly less cycles than in the bottom plot, and $T_2$ uses the same number in both cases, but is run at $f_1$ in the first case, and at $f_2$ in the second one.}
\end{center}
\end{figure}

First, we have to show that in the range $[0, z_i]$, $\dfrac{w_i}{z_{i+1}-t}\le f_M$. As this function is an increasing function of $t$, we just need to consider the maximal value we need:
\begin{eqnarray*}
\dfrac{w_i}{z_{i+1}-z_i} &=& \dfrac{w_i}{D-\dfrac{1}{f_M}\sum\limits_{k=i+1}^N w_k - \left(D-\dfrac{1}{f_M}\sum\limits_{k=i}^N w_k\right)}\\
	&=& \dfrac{w_i}{\frac{1}{f_M} w_i} = f_M\\
\end{eqnarray*}

~\\
\textbf{Initialization.}~~For the initialization, we consider $T_1$. Clearly, as the execution length is not taken into account for the frequency selection, the worst case occurs when $T_1$ uses $w_1$ cycles. As $T_1$ starts at time $0$, we have
$$
\tau_1 = \frac{w_1}{S_1(0)}.
$$
As $S_1(t)\ge \dfrac{w_1}{z_2-t}$ by hypothesis, we have

$$
\tau_1 \le \dfrac{w_1}{\frac{w_1}{z_2}} = z_2.
$$

$T_1$ ends then no later than $z_2$ in any case. Similarly, we have that if $S_1(t) < \dfrac{w_1}{z_2-t}$, $\tau_1 >z_2$, and we cannot guarantee that $T_1$ finishes no later than $z_2$

~\\
\textbf{Induction.}~~Let us now consider $T_i$, with $i>1$. We know by induction that $T_{i-1}$ finished its execution between time $0$ and time $z_i$. Let $\theta$ be this end time. Knowing that task $T_i$ starts at $\theta$, the worst case for $T_i$ is to use $w_i$ cycles. The worst end time of $T_i$ is then

$$
\tau_i = \theta + \frac{w_i}{S_i(\theta)}
$$
with $\theta \in [0, \tau_{i-1}=z_i]$.

Then, as $S_i(t) \ge \dfrac{w_i}{z_{i+1}-t}$ (which is possible, because we have just shown that the right hand side is not higher than $f_M$ in the range we have to consider), we have
$$
\tau_i = \theta + \frac{w_i}{S_i(\theta)} \le \theta + \frac{w_i}{\frac{w_i}{z_{i+1}-\theta}} =\theta + z_{i+1}-\theta  = z_{i+1}.
$$

We then have that if $S_i(t) \ge \dfrac{w_i}{z_{i+1}-t}$, task $T_i$ finishes always no later than $z_{i+1}$, and then, as a consequence, that any task finishes no later than $z_{N+1}=D$.

Symmetrically, we can show also that if $S_i(t) < \dfrac{w_i}{z_{i+1}-t}$, then $\tau_i$ is higher than $z_{i+1}$, and then $\tau_N$ is higher than $D$, and the system is not schedulable.
\end{proof}

Remark that the expedience hypothesis is a little bit too strong. It would be enough to require that $T_i$ never waits intentionally later than $z_i$. $T_1$ doesn't even have to start at time $0$, as soon as it starts no later that $z_1$. With this hypothesis, the initialization would be: in the worst case, $T_1$ would start at time $\theta$, somewhere between $0$ and $z_1$ and use $w_1$ cycles. In this case, it would end at
$$
\tau_1 = \theta + \frac{w_1}{S_1(\theta)} \le \theta+ \frac{w_1}{\frac{w_1}{z_2-\theta}} = z_2
$$
and we know that the CPU can be set to the speed $\frac{w_1}{z_2-\theta}$, which is not higher than $f_M$ because $\theta$ is in $[0, z_1]$.

\begin{definition}
   We denote by $\mathcal{L}_i(t)$ the schedulability limit, or
   $$\mathcal{L}_i(t) = \dfrac{w_i}{z_{i+1}-t}$$
 where
 $$z_i = D-\frac{1}{f_M}\sum_{k=i}^N w_k.$$
\end{definition}

An example of such schedulability limits is given in Figure~\ref{fig:lim}, with four tasks, and a maximum frequency of 1000MHz.

\begin{figure}[ht]
\begin{center}
\includegraphics[width=0.75\linewidth]{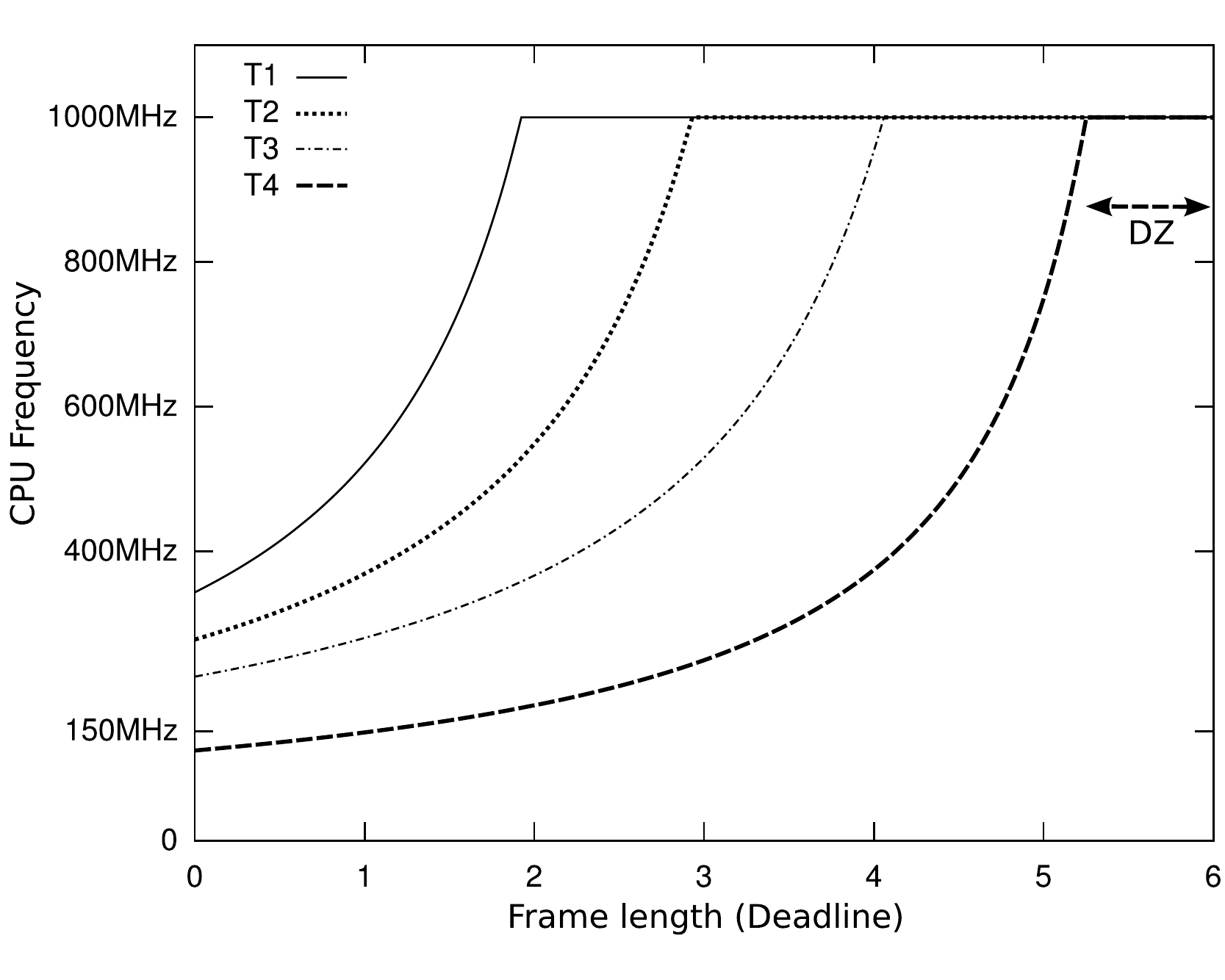}
\caption{\label{fig:lim}Set of limit functions $\mathcal{L}_i(t)$, for an example of 4 tasks. DZ represents the Danger Zone of $T_4$.}
\end{center}
\end{figure}

\subsection{Discrete Limit}
The closest scheduling functions set to the limit is
$$
S_i(t) = \min\left\{f \in \{f_1, \dots, f_N\} :  f \ge \mathcal{L}_i(t) \right\}. 
$$
Informally, we could write this function $S_i(t) = \left\lceil\dfrac{w_i}{z_{i+1} - t}\right\rceil$, where $\lceil w \rceil$ stands for ``\textit{the smallest available frequency not lower than} $x$''.
This function varies as a discrete hyperbola between $\left\lceil\dfrac{w_i}{z_{i+1}}\right\rceil$ and 
$$\left\lceil\dfrac{w_i}{z_{i+1} - z_i}\right\rceil=\left\lceil\dfrac{w_i}{\frac{w_i}{f_M}}\right\rceil = \left\lceil f_M\right\rceil = f_M.$$

This function is however in general not very efficient: $T_1$ is run at the slowest frequency allowing to still run the following jobs in the remaining time. But then, $T_1$ is run very slowly, while $\{T_2, \dots, T_N\}$ have a pretty high probability to run at a high frequency. A more balanced frequency usage is often better.

This strategies actually corresponds to the Greedy technique (DPM-G) described by Moss\'{e} et al. \cite{Mosse00}, except that they consider continuous speeds.

Building such a function is very easy, and is in $\mathcal{O}(M)$ for each task, with the method given by Algorithm~\ref{alg:limit}. We mainly need to be able to inverse $\mathcal{L}$: $\mathcal{L}^{-1}_i(f) = z_{i+1} - \frac{w_i}{f}$.

\begin{algorithm}[ht]
\caption{\label{alg:limit} Building {\sffamily Limit}, worst case scheduling functions. $(a)^+$ means $\max\{0, a\}$.}
$z\gets D$\;

\ForEach{$i \in \{N, \dots, 1\}$}{
	$S_i \adds (0, f_1)$\;
	
	\ForEach{$j \in \{2, \dots, M\}$}{
		$S_i \adds \big( \left(z - \dfrac{w_i}{f_{j-1}}\right)^+ , f_j\big)$\;
	}
	$z \gets z - \frac{w_i}{f_M}$\;
}
\end{algorithm}

In the following, this strategy is named as {\sffamily Limit}.

\subsection{Checking the schedulability}

Provided a set of scheduling functions $\{S\}$, checking its schedulability is pretty simple. As we know that the limit function is non decreasing, we just need to check that each step of $S_i$ is above the limit. This can be done with the following algorithm.

\begin{algorithm}[ht]
\caption{\label{alg:check}Schedulability check}
$z\gets D$\;

\ForEach{$i \in \{N, \dots, 1\}$}{
	
	\ForEach{$k \in \{2, \dots, \mid S_i \mid\}$}{
		\If{$S_i[k-1].f < \dfrac{w_i}{z - S_i[k].t}$}{
			\Return \KwSty{false}\;
		}
	}
	$z \gets z - \frac{w_i}{f_M}$\;
}
\Return \KwSty{true}\;
\end{algorithm}

This check can then be performed in $\mathcal{O}\left(\sum_{i=1}^N \mid S_i \mid\right)$ which, is $S_i$ is non decreasing (which is almost always the case) is lower than $\mathcal{O}(N\times M)$.

This test can be used offline to check the schedulability of some method or heuristic, but can also be performed as soon as some parameter change has been detected. For instance, if the system observes that a task $T_i$ used more cycles than its (expected) WCEC $w_i$, the test could be performed with the new WCEC in order to see if the current set of $S$ functions can still be used. Notice that we only need to check tasks between $1$ and $i$, because the schedulability of tasks in $\{i+1, \dots, N\}$ does not depend upon $w_i$.
See Section~\ref{sec:ccl} about future work for more details.

\subsection{Using Schedulability Condition to Discretize Continuous Methods}
\label{sec:discret}
\begin{figure}[ht]
\begin{center}
\includegraphics[width=0.80\linewidth]{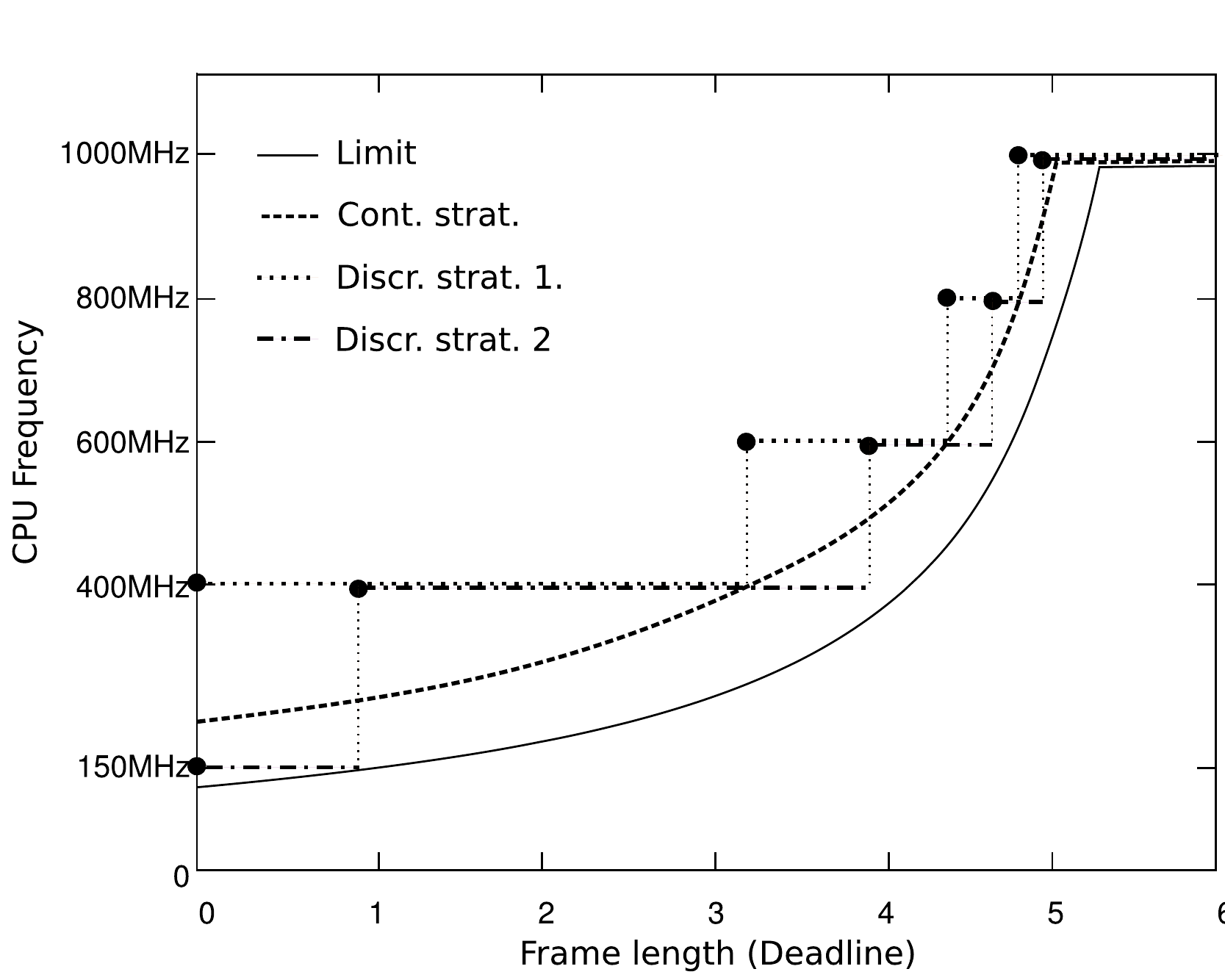}
\end{center}
\caption{Two different ways of discretizing a continuous strategy: {\sffamily Discr.\;strat.\;1} rounds up to the first available frequency. {\sffamily Discr.\;strat.\;2} (our proposal) uses the closest available frequency, taking the limit into account. {\sffamily Limit} is the strategy described by Algorithm~\ref{alg:limit}.}
\end{figure}

There are mainly two ways of building a set of $S$-functions for a given system. The first method consists in considering the problem with continuous available frequencies, and by some heuristic, adapting this result for a discrete speeds system. The second method consists in taking into account from the beginning that there are only a limited number of available speeds.
The second family of methods has the advantage of being usually more efficient in terms of energy, but the disadvantage of being much more complex, requiring a non negligible amount of computations or memory. This is not problematic if the system is very stable and its parameters do not change often, but as soon as some on-line adaptation is eventually required, heavy and complex computations cannot be performed anymore.

In the first family, the heuristic usually used consists in computing a continuous function $\mathcal{S}^c_i(t)$ which is build in order to be schedulable, and to obtain a discrete function by using for any $t$ the smallest frequency above $\mathcal{S}^c_i(t)$, or $S_i(t) = \lceil \mathcal{S}^c_i(t)\rceil$. However, this strategy is often pessimistic. But so far, there were no other method in order to ensure the schedulability. This assertion is not valid anymore, because we provided in this paper a schedulability condition which can be used.

The main idea is, instead of using the smallest frequency above $\mathcal{S}^c_i(t)$, to use the closest frequency to $\mathcal{S}^c_i(t)$, and, if needed, to round this up with the schedulability limit $\mathcal{L}_i(t)$. In other words, we will use:
$$
S_i(t) = \max \{ \lceil \mathcal{S}^c_i(t) \rfloor, \lceil \mathcal{L}_i(t) \rceil \}.
$$
The advantage of this technique is that we have more chance to be closer to the continuous function (which is often optimal in the case of continuous CPU). However, both techniques (ceiling and closest frequency) are approximations, and none of them is guaranteed to be better than the other one in any case. As we will show in the experimental section, there are systems in which the classical discretization is better, but there are also many cases where our discretization is better.

Algorithm~\ref{alg:closest} shows how step functions can be obtained. For each task, computing its function is in $\mathcal{O}(M \times A)$, where $A$ is the complexity of computing $\mathcal{S}^{-1}_i(f)$. According to the kind of continuous method we use, $A$ can range between $1$ (if ${\mathcal{S}^c}^{-1}_i(f)$ has a constant closed form) and $\log(D/\varepsilon) \times B$, with a binary search, where $\varepsilon$ is the desired precision, and $B$ the complexity of computing $\mathcal{S}^c_i(t)$.
\begin{algorithm}[ht]
\caption{\label{alg:closest}Algorithm computing the closest stepfunction to $\mathcal{S}^c_i(\cdot)$, respecting the schedulability limit $\mathcal{L}_i(\cdot)$. $(a)^+$ stands for $\max\{0, a\} $.}
\ForEach{$i \in \{N, \dots, 1\}$}{
	$S_i \adds (0, f_1)$\;
	
	\ForEach{$j \in \{2, \dots, M\}$}{
		$f \gets (f_{j-1}+f_j)/2$\; \\
		$t \gets \KwSty{min} \{ {\mathcal{S}^c}^{-1}_i(f), \mathcal{L}^{-1}_i(f_{j-1})\}$\; \\
		$S_i \adds ((t)^+, f_j)$\;
	}
}
\end{algorithm}

Actually, computing the closest frequency amongst $\{f_1, f_2, \dots, f_M\}$ roughly boils down to compute the round up frequency amongst the set $\{\frac{f_1+f_2}{2}, \frac{f_2+f_3}{2}, \dots, \frac{f_{M-1}+f_M}{2}\}$. Then, the range corresponding to $\frac{f_1+f_2}{2}$ is mapped onto $f_2$, etc.
In Algorithm~\ref{alg:closest}, if we simply use $f_{j-1}$ instead of $f$, we obtain the classical round up operation.

\section{Model Extensions}
\label{sec:ext}
\subsection{Frequency Changes Overhead}
\label{sec:overhead}
Our model allows to easily take the time penalty of frequency changes into account. Let $P_T(f_i, f_j)$ be the time penalty of changing from $f_i$ to $f_j$. This means that once the frequency change is asked (usually, a special register has been set to some predefined value), the processor is ``idle'' during $P_T(f_i, f_j)$ units of time before the next instruction is run. We assume that the worst time overhead is when the CPU goes from $f_1$ to $f_M$. We denote for this $P_T^M = \max_{i,j} P_T(f_i, f_j) = P_T(f_1, f_M)$.

Notice that this model is rather pessimistic: on modern DVS CPUs, the processor does not stop after a change request, but still run at the old frequency for a few cycles before the change becomes effective. However, even if the processor never stops, there is still a penalty, but the time penalty is negative when the speed goes down (because the job will be finished sooner than if the frequency change had been performed before it started). Then as a first approximation, we could consider that negative penalties compensate positive penalties. But this approximation does not hold for energy penalties, because all of them are obviously positive.

We want also to take the switching time before jobs into account, even if there is no frequency change (we assume that the job switching time is already taken into account in $P_T$). Let $S_T(f_i)$ be the switching time when the frequency is $f_i$, and is not changed between two consecutive jobs. Again, let $S_T^M$ denote $S_T(f_M)$. Usually, we have $S_T(f_i)<S_T(f_j)$ if $f_i>f_j$. We made here the simplifying hypothesis that the switching time is job independent, which is an approximation since this time usually depends upon the amount of used memory. However, in our purpose, we only need to consider an upperbound of this time.

As before, we know that $T_N$ must start no later than $D-\frac{w_N}{f_M}$. If $T_N$ starts at this limit (and even before), the selected frequency must be $f_M$. Then we could have two situations:
\begin{itemize}
   \item Best case: the previous tasks $T_{N-1}$ was already running at $f_M$. Then $T_{N-1}$ needs to finish before the start limit for $T_N$, minus the switching time, then $D-\frac{w_N}{f_M} - S_T^M$;
   \item Worst case: the previous tasks $T_{N-1}$ was not running at $f_M$, we need then to change the frequency. In the worst case, the time penalty will be $P_T^M$. $T_{N-1}$ needs then to finish no later than $D-\frac{w_N}{f_M} - P_T^M$.
\end{itemize}

The first limit is then a necessary condition, and the second, a sufficient condition to ensure the schedulability of $T_N$. Similarly, we can see that $T_i$ must start before $z_i^n$ to ensure the schedulability of itself and any subsequent task (necessary condition), and this schedulability is ensured (sufficient condition) if $T_i$ starts before $z_i^s$, where $z_i^n$ and $z_i^s$ are defined as:
$$
z_i^n = D - \frac{1}{f_M} \sum_{k=i}^N w_k - (N-i+1) S_T^M = z_i - (N-i+1) S_T^M
$$
and
$$
z_i^s = D - \frac{1}{f_M} \sum_{k=i}^N w_k - (N-i+1) P_T^M = z_i - (N-i+1) P_T^M
$$
We can then provide two schedulability conditions:
\begin{itemize}
   \item Necessary condition: $S_i(t) \ge \dfrac{w_i}{z_{i+1}^n - t}$;
   \item Sufficient condition: $S_i(t) \ge \frac{w_i}{z_{i+1}^s - t}$.
\end{itemize}

Algorithm~\ref{alg:closest} can easily be adapted using those conditions. We use then $\mathcal{L}_i(t) = \dfrac{w_i}{z_{i+1}^s - t}$.

\subsection{Soft Deadlines}
If we want to be a little bit more flexible, we could possibly consider soft deadlines, and adapt our schedulability condition consequently. The main idea is to not consider the WCEC, but to use some percentile: if $\kappa_i(\varepsilon)$ is such that $\mathbb{P}[c_i<\kappa_i(\varepsilon)] \ge 1-\varepsilon$, where $c_i$ is the actual number of cycles of $T_i$, we can use $\kappa_i(\varepsilon)$ as a worst case execution time.

However, it seems to be almost impossible to compute analytically the probability of missing a deadline with this model. It would boil down to compute $\mathbb{P}[E_1+E_2+E_3+...+E_N]$ where $E_i$ represents the execution time of jobs of task $T_i$. $E_i$ depends then upon the job length distribution, but also upon the speed at which $T_i$ is run, which depends upon the time at which $T_{i-1}$ ends ... which depends upon the time $T_{i-2}$ ended, and so on.
As $E_i$'s are not independent, it seems then that we cannot use the central limit theorem.

If we accept an approximation of the failure probability, we could do in the following way. Let $C_i$ be the random variable giving the number of cycles of $T_i$, and $\mathbb{C} = \sum_i C_i$. Let $\mathbb{W} = \sum_i w_i$ be the maximal value of $\mathbb{C}$ (the frame worst case execution cycle). Let $\mathbb{C}^\varepsilon=\min_c \{ \mathbb{P}[\mathbb{C}<c]>1-\varepsilon\}$.

We assume that using the deadline $D \dfrac{\mathbb{W}}{\mathbb{C}^\varepsilon}$ will allow to respect deadlines with a probability close to $1-\varepsilon$. Those propositions are only heuristics, and should require more work, both analytic and experimental.

\section{Experimental Results}

\newcommand{\figwidth}{0.3\linewidth}
\label{sec:experim}
\begin{figure*}[h!t!]
\begin{center}
\includegraphics[width=\figwidth]{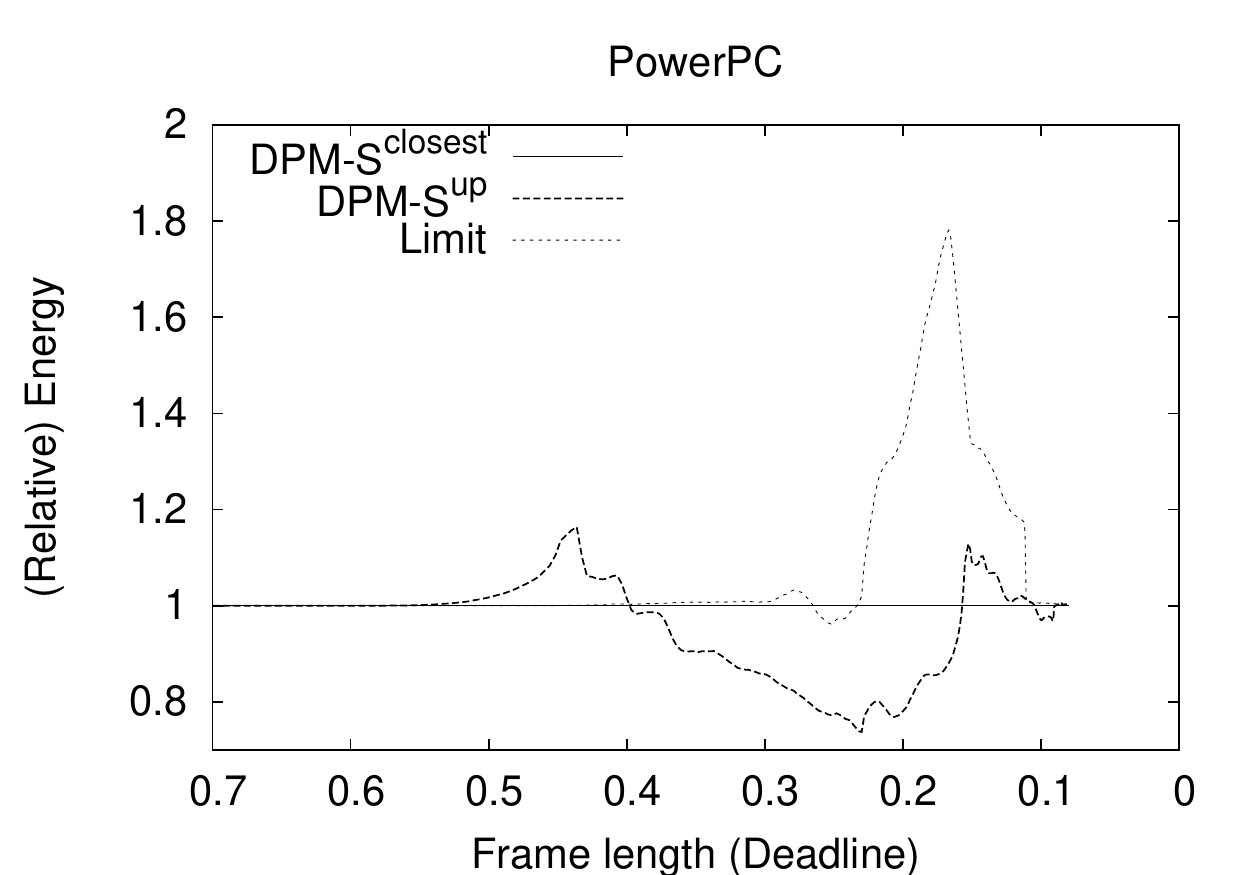}
\includegraphics[width=\figwidth]{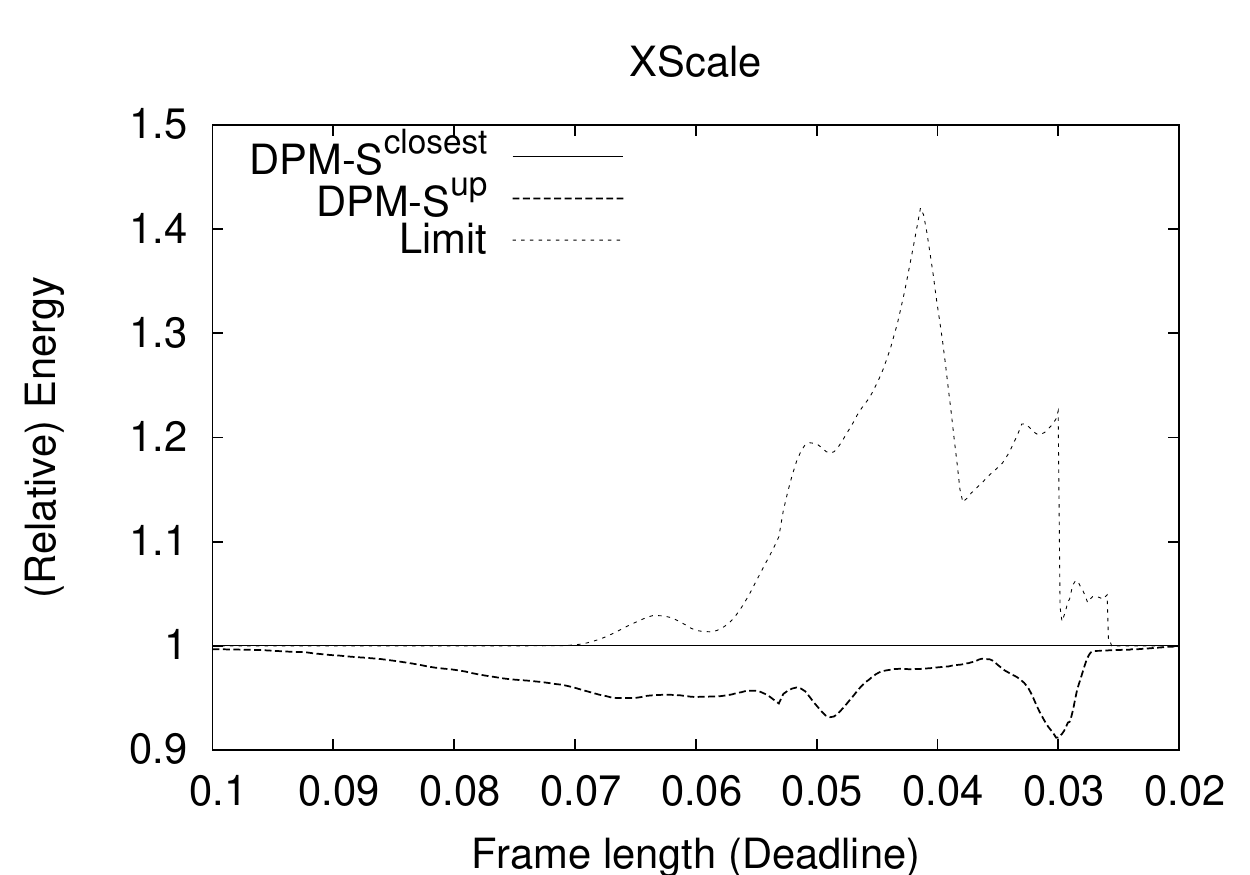}
\includegraphics[width=\figwidth]{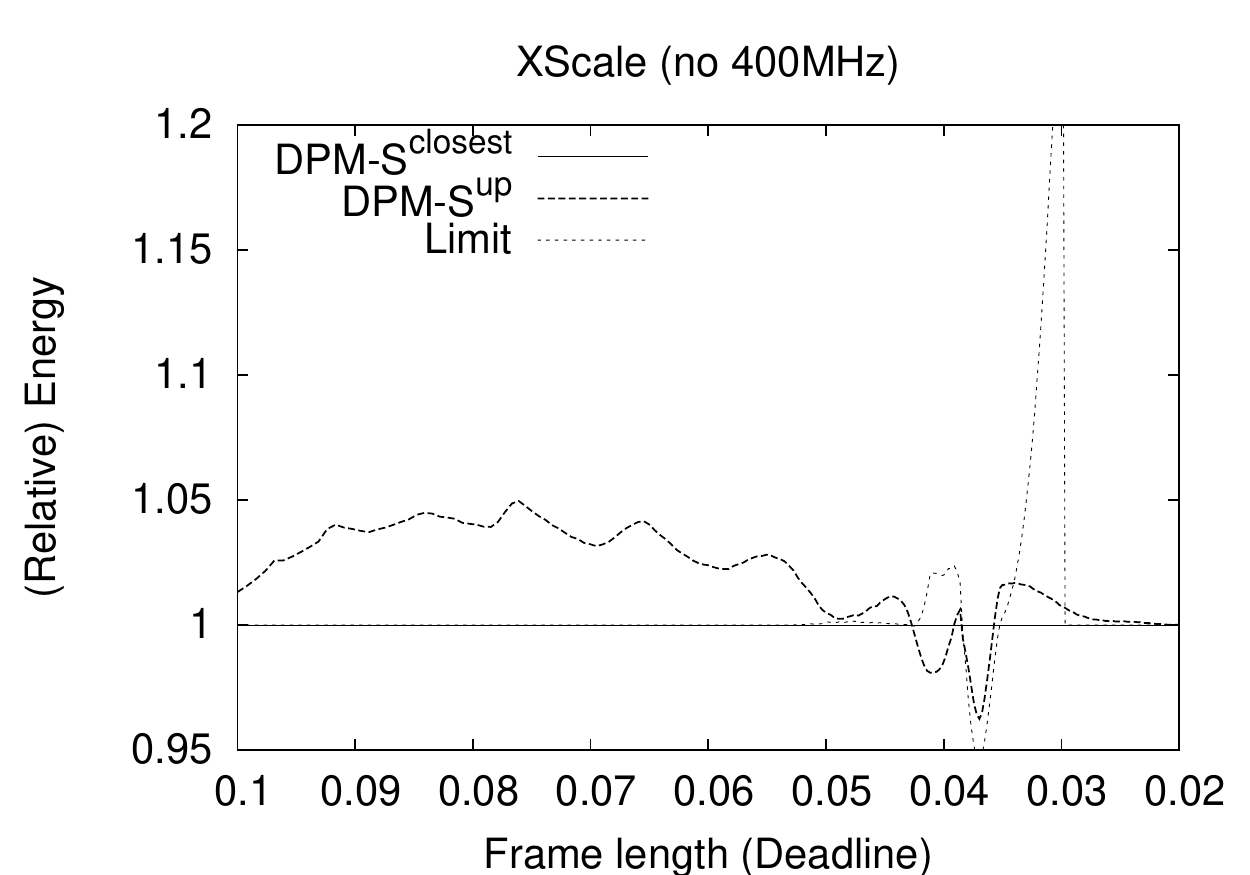}
\end{center}
\caption{\label{fig:Unif-DPMS}Energy consumption relative to {\sffamily DPM-S$^{\text{closest}}$}, for a set of 12 tasks with uniformly distribution.}
\end{figure*}

\begin{figure*}[h!t!]
\begin{center}
\includegraphics[width=\figwidth]{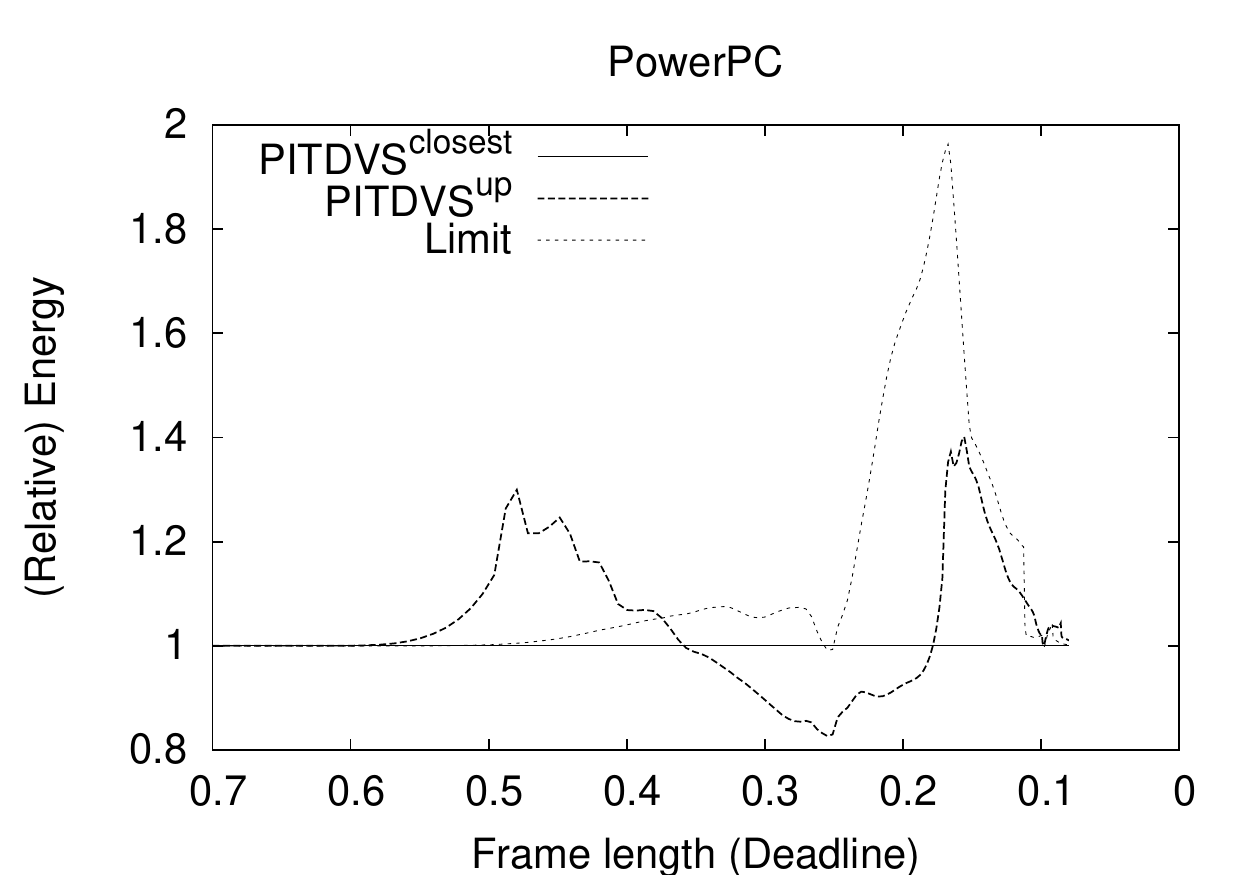}
\includegraphics[width=\figwidth]{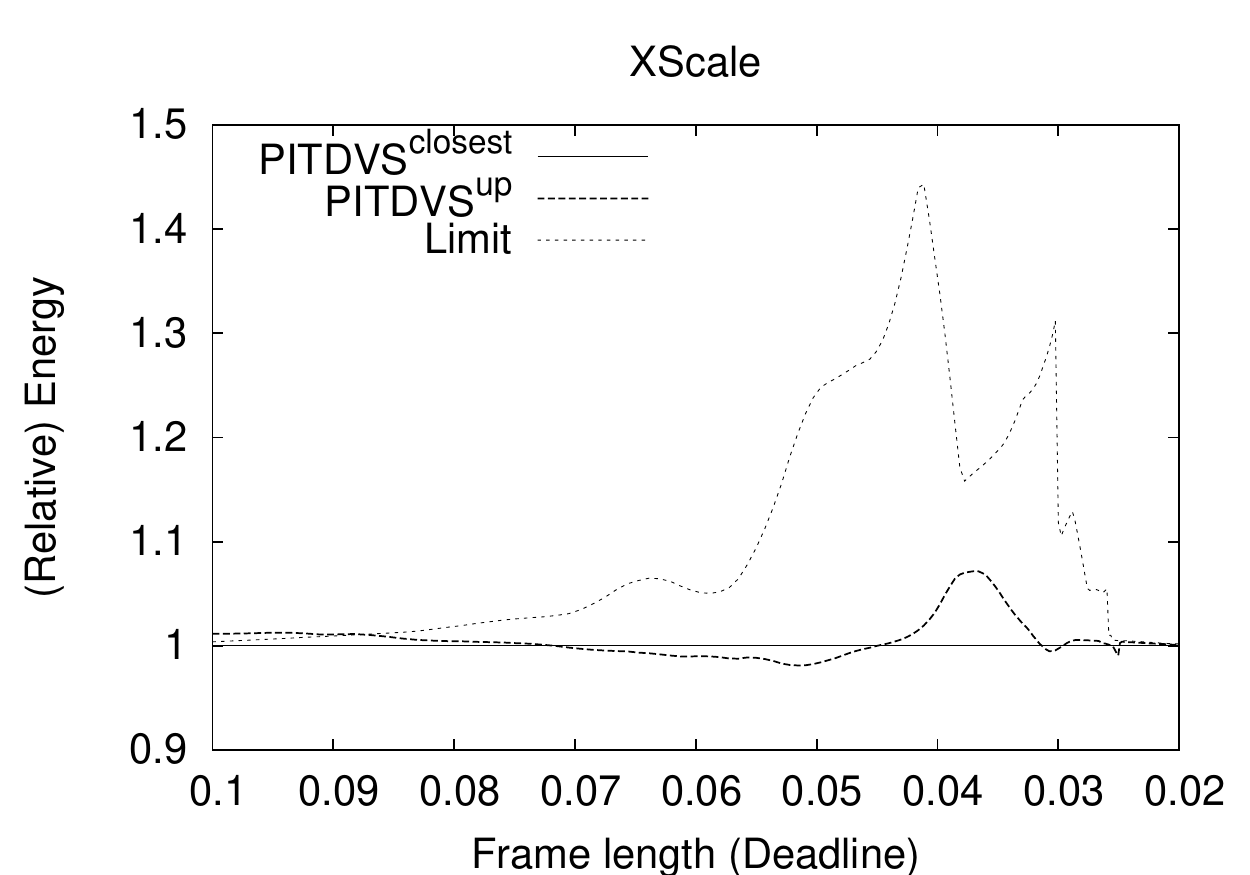}
\includegraphics[width=\figwidth]{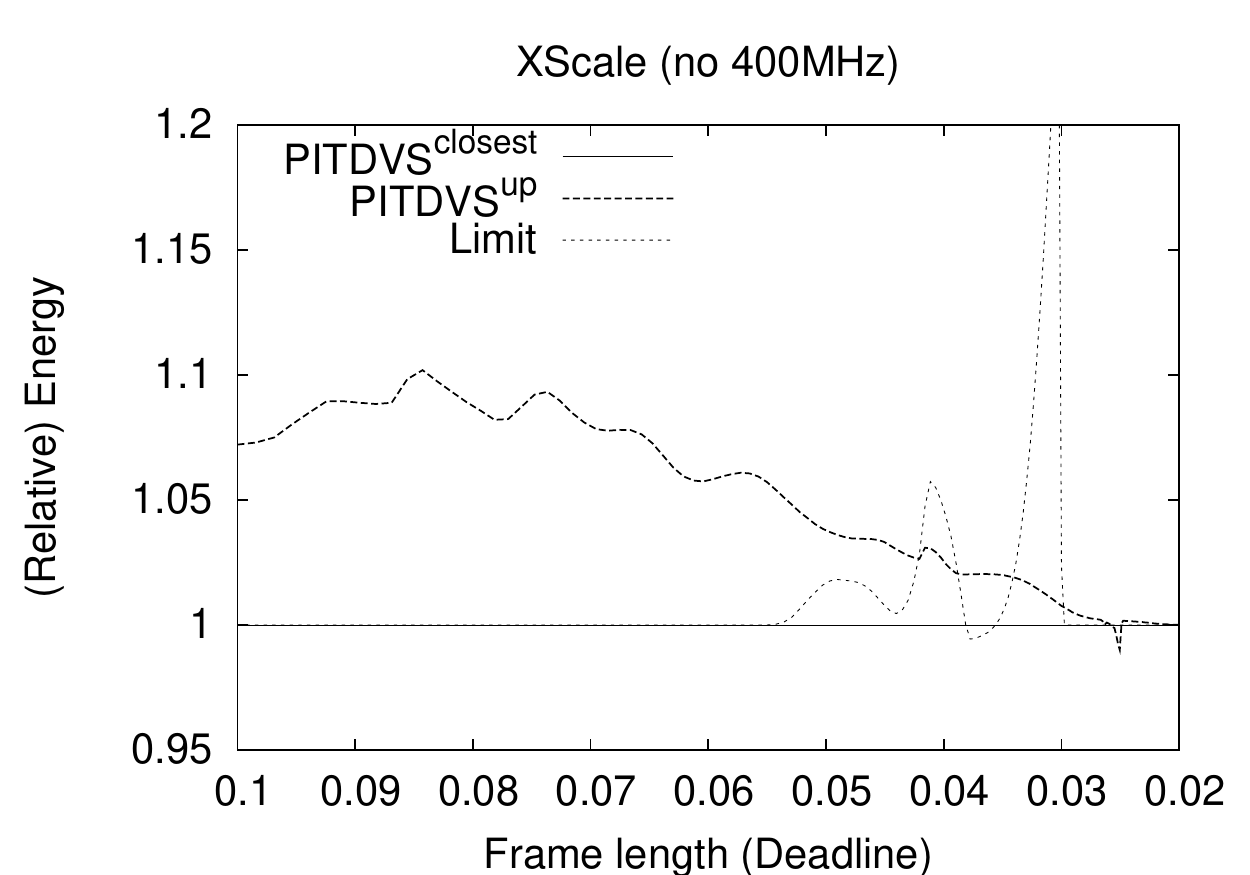}
\end{center}
\caption{\label{fig:Unif}Energy consumption relative to {\sffamily PITDVS$^{\text{closest}}$}, for a set of 12 tasks with uniformly distribution.}
\end{figure*}

\begin{figure*}[ht]
\begin{center}
\includegraphics[width=\figwidth]{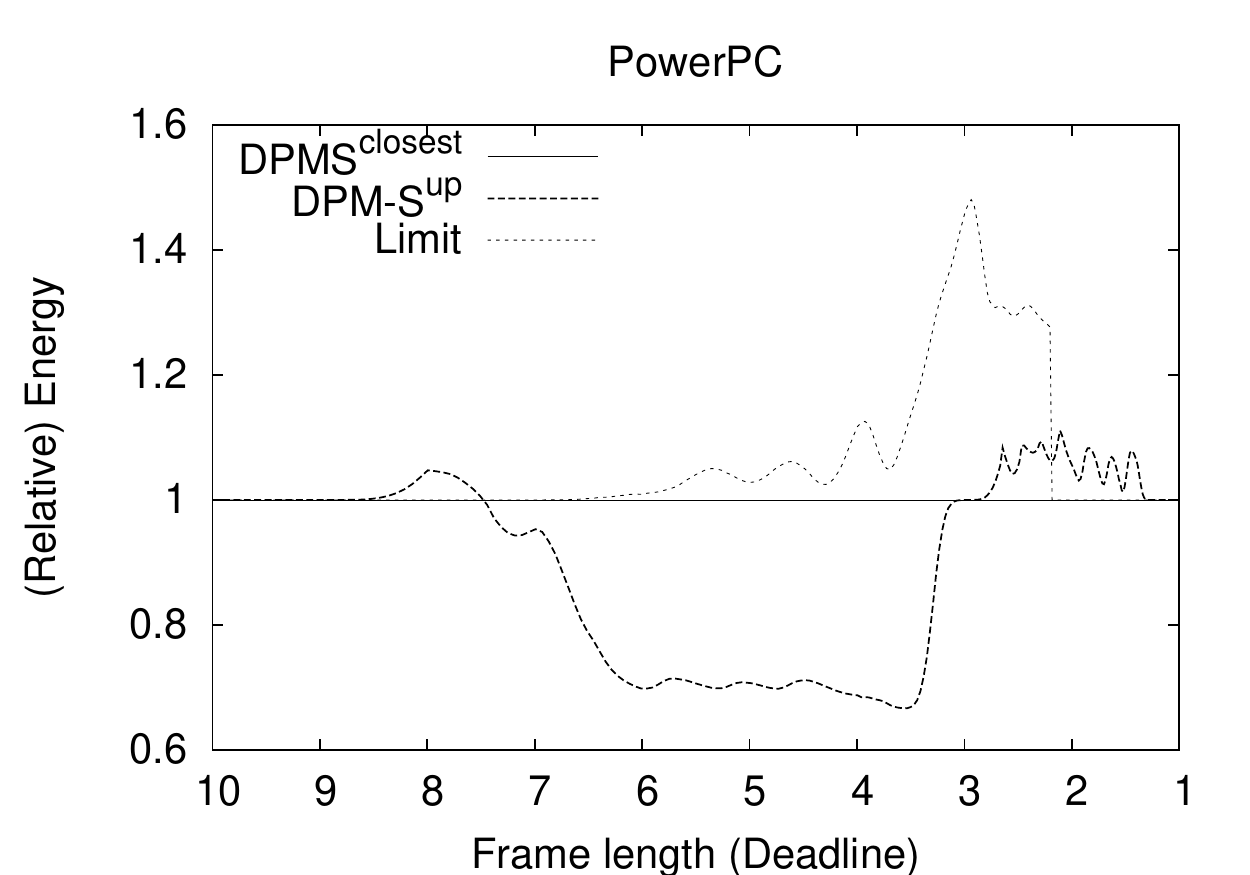}
\includegraphics[width=\figwidth]{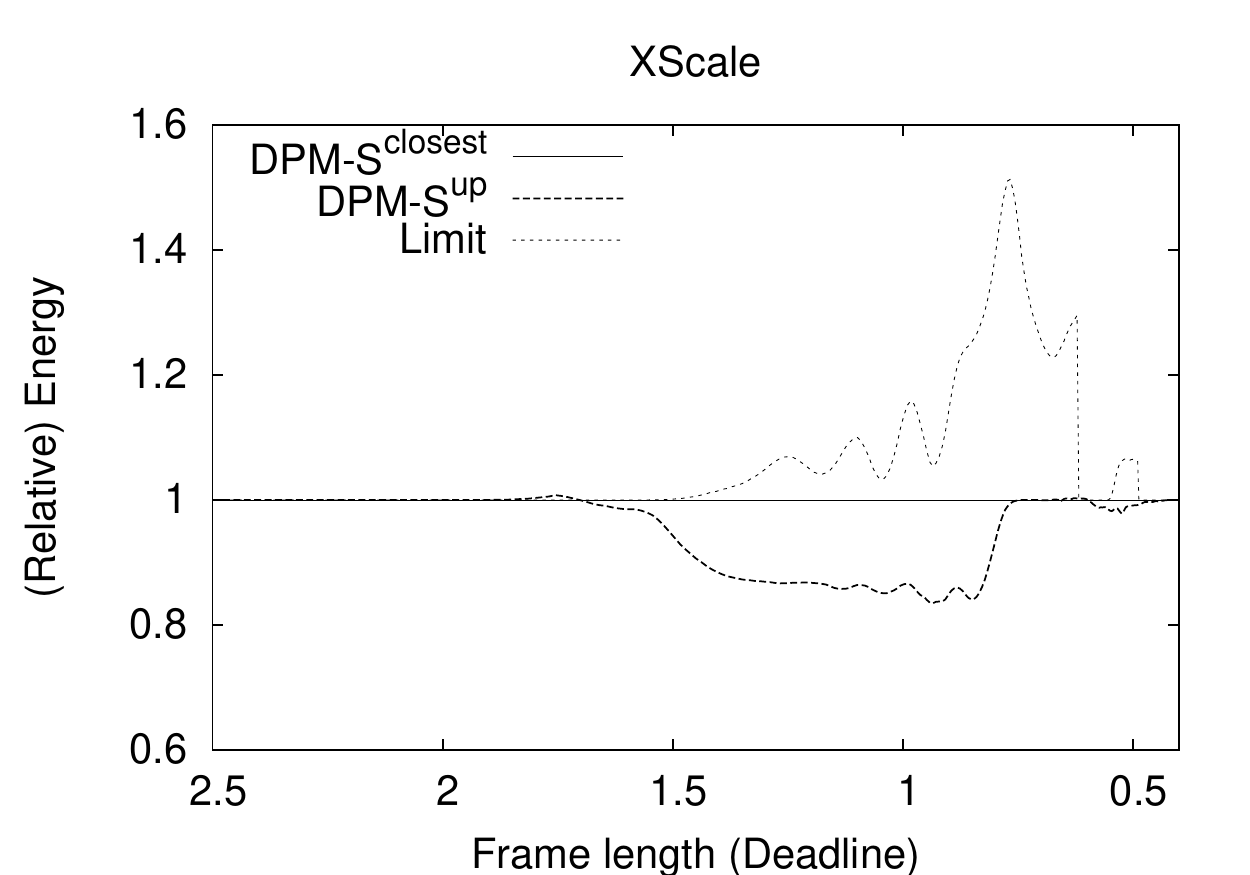}
\includegraphics[width=\figwidth]{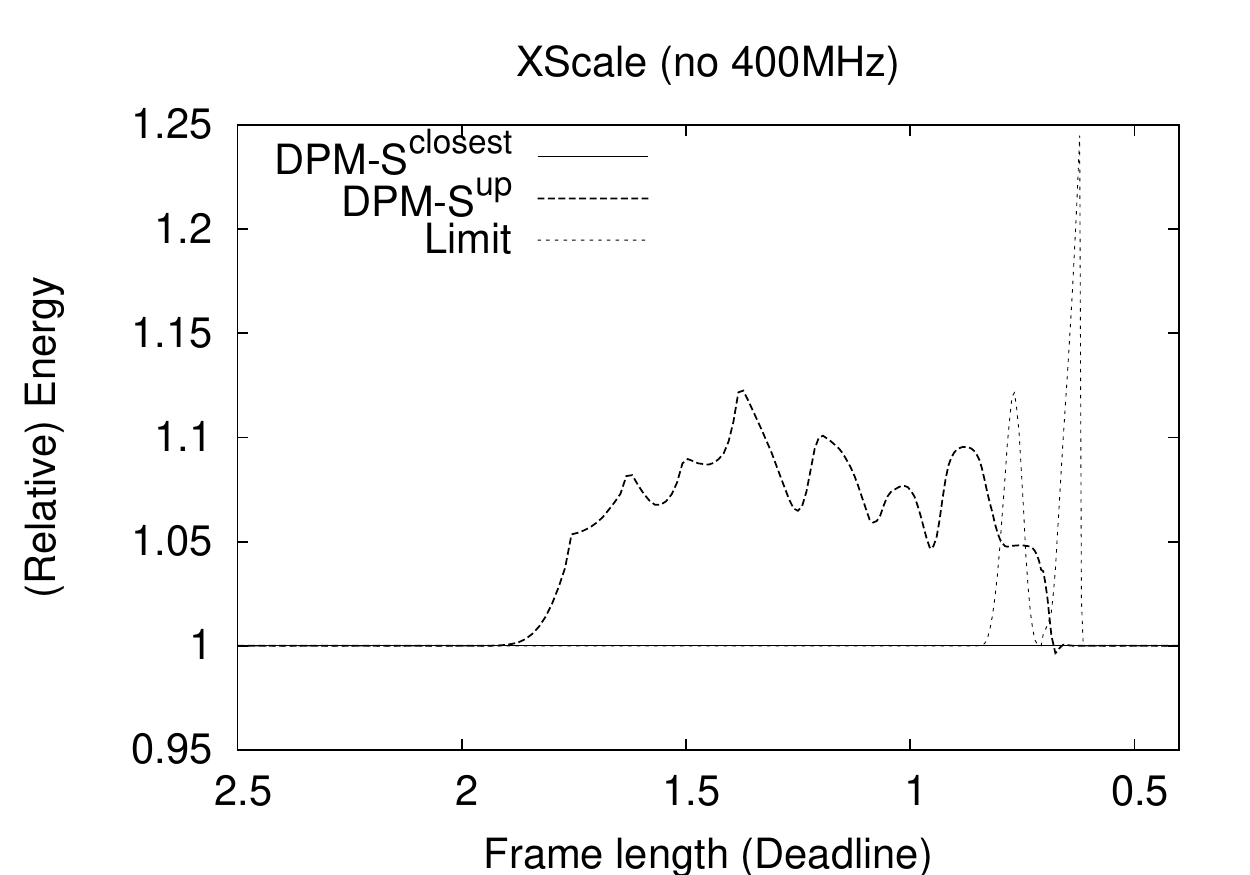}
\end{center}
\caption{\label{fig:Div-DPMS}Energy consumption relative to {\sffamily DPM-S$^{\text{closest}}$}, for a set of 8 tasks distributed as shown in Figure \ref{fig:DivDistr}.}
\end{figure*}

\begin{figure*}[ht]
\begin{center}
\includegraphics[width=\figwidth]{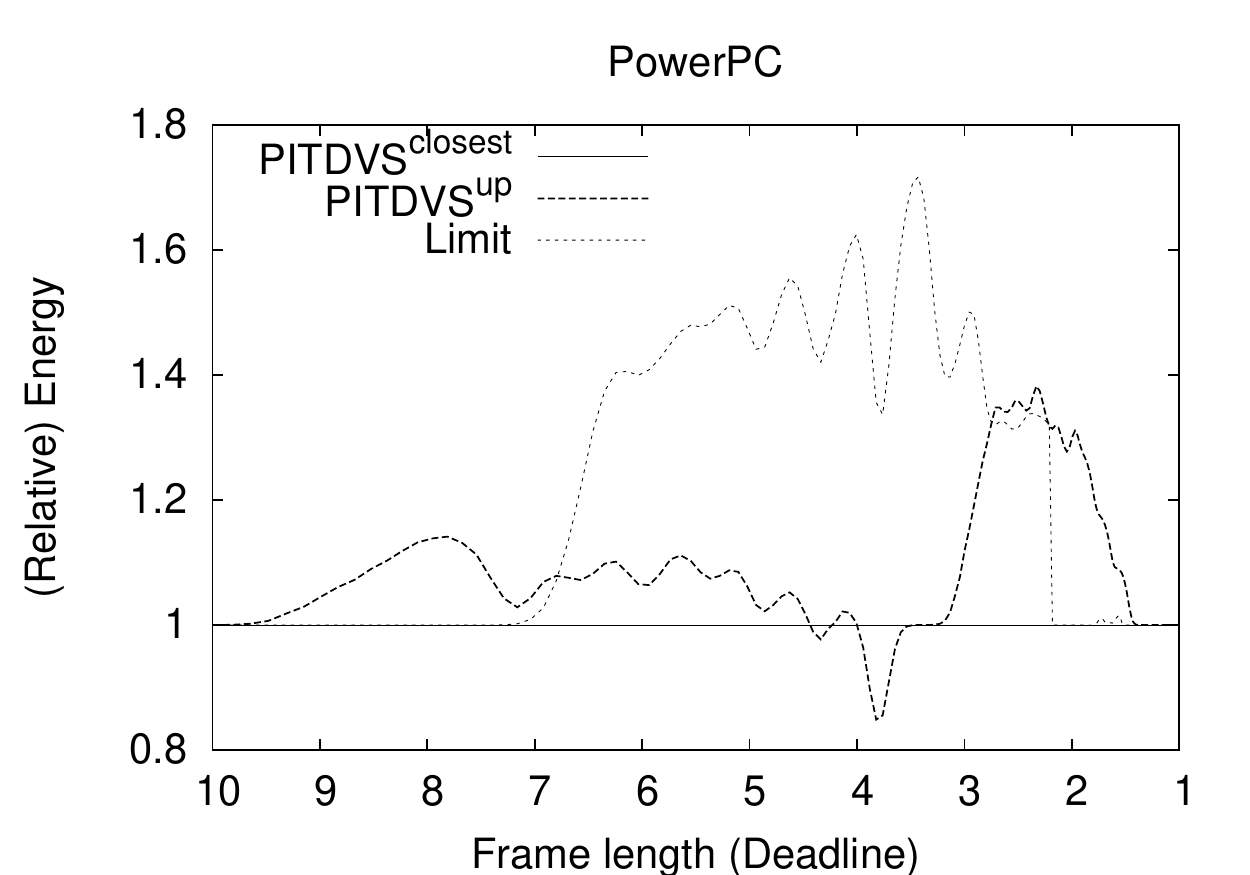}
\includegraphics[width=\figwidth]{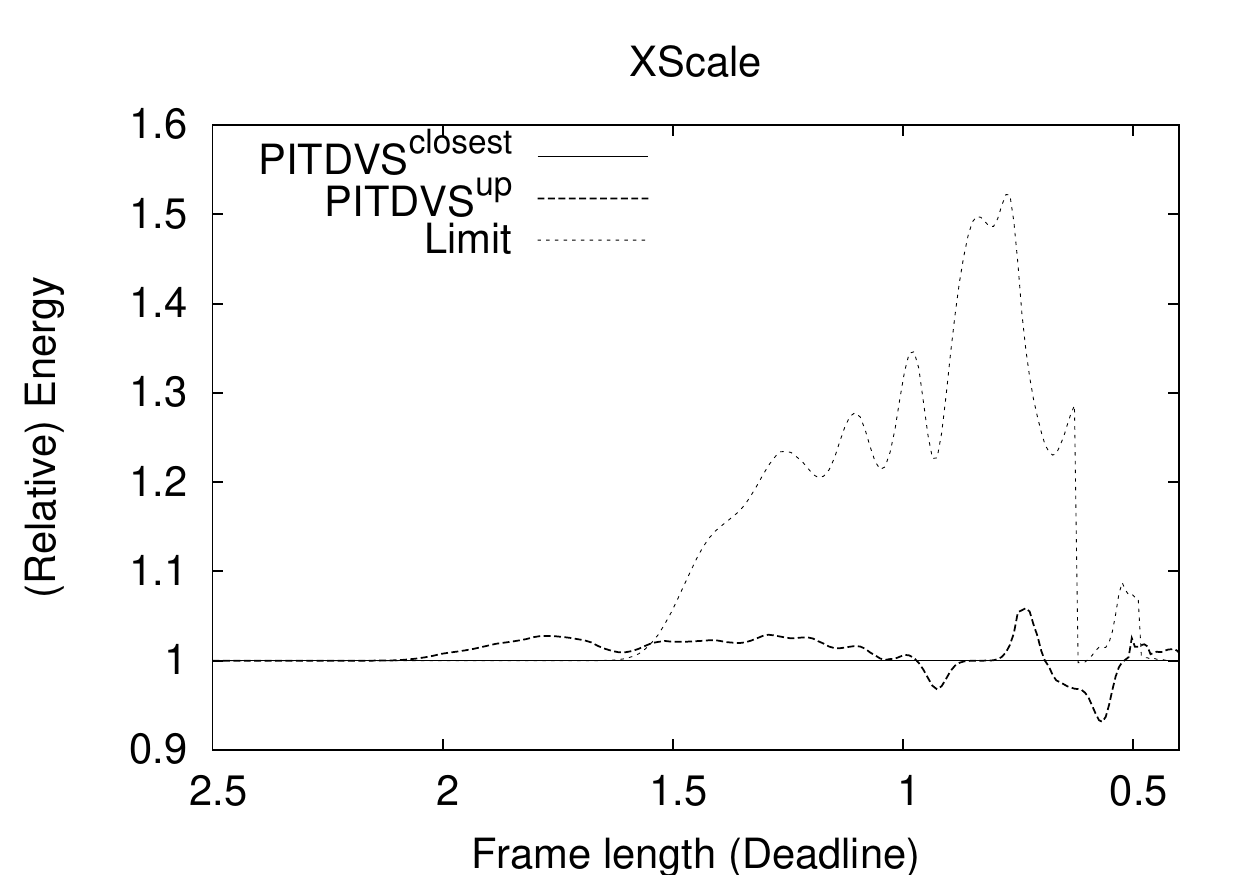}
\includegraphics[width=\figwidth]{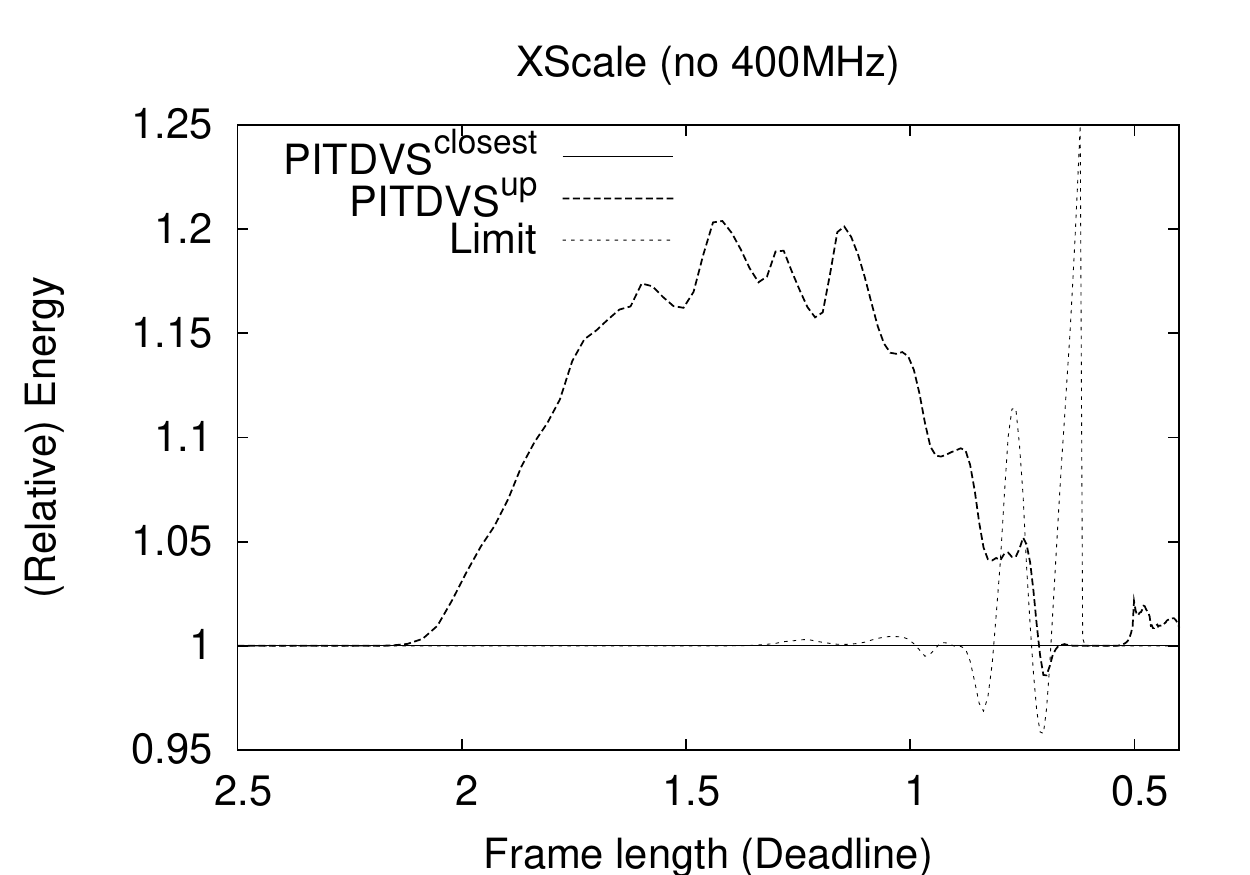}
\end{center}
\caption{\label{fig:Div}Energy consumption relative to {\sffamily PITDVS$^{\text{closest}}$}, for a set of 8 tasks distributed as shown in Figure \ref{fig:DivDistr}.}
\end{figure*}

In order to evaluate the advantage of using a ``closest'' approach instead of an ``upper bound'' approach, we applied it on two methods. The first is one described by Moss\'{e} et al. in \cite{Mosse00}, and is called DPM-S (Dynamic Power Management-Statistical), and the second one is described by described by Xu, Melhem and Moss\'e~\cite{Xu07b}, called PITDVS (Practical Inter-Task DVS).

\subsection{DPM-S}
The method DPM-S described in \cite{Mosse00} bets that the next jobs will not need more cycles than their average, and compute then the speed making this assumption when a job starts. Of course, the schedulability limit is also taken into account. In their paper, the authors consider that they can use any (normalized) frequency between $0$ and $1$. In order to apply this method on a system with a limited number of frequencies, we can either round them up, or use or ``closest'' approach. They don't take frequency change overheads into account, but according to what we claimed hereabove, those overheads are easy to integrate.

We compute now the two following step functions in this way, where $avg_i$ stands for the average number of cycles of $T_i$: in  Algorithm~\ref{alg:closest} adapted to take frequency changes overhead into account (cf Section~\ref{sec:overhead}),
\begin{itemize}
   \item {\sffamily DPM-S$^{\text{up}}$}: we replace $\mathcal{S}_i^{-1}$ by
   \begin{equation}\label{eq:dpmsup}
      D - \frac{\sum_{j=i}^N avg_i}{f_{j-1}};
   \end{equation}

   \item {\sffamily DPM-S$^{\text{closest}}$}: we replace $\mathcal{S}_i^{-1}$ by
   \begin{equation}\label{eq:dpmsclosest}
      D - \frac{\sum_{j=i}^N avg_i}{f}.
   \end{equation}
\end{itemize}

\subsection{PITDVS}
The second method we consider, by Xu, Melhem and Moss\'e~\cite{Xu07b}, is called PITDVS (Practical Inter-Task DVS), and aims at patching OITDVS (Optimal Inter-Task DVS \cite{Xu05}), an optimal method for ideal processors (with a continuous range of available frequencies). They apply several patches in order to make this optimal method usable for realistic processors. They start by taking speed change overhead into account, then they introduce maximal and minimal speed (OITDVS assumes speed from 0 to infinity), and finally, they round up the $S$-function to the smallest available frequency. It is in this last patch that we apply our technique.
Using the $\beta_i$ value described in \cite{Xu07b} (representing the aggressiveness level), we compute the step functions in the following way: in  Algorithm~\ref{alg:closest} adapted to take frequency changes overhead into account (cf Section~\ref{sec:overhead}),
\begin{itemize}
   \item {\sffamily PITDVS$^{\text{up}}$} (in \cite{Xu07b}): we replace $S_i^{-1}$ by
   \begin{equation}\label{eq:pitdvsup}
      D - P_T \times(N-i) - \frac{w_i}{\beta_i f_{j-1}};
   \end{equation}

   \item {\sffamily PITDVS$^{\text{closest}}$} (our adaptation): we replace $S_i^{-1}$ by
   \begin{equation}\label{eq:pitdvsclosest}
      D - P_T \times(N-i) - \frac{w_i}{\beta_i f}.
   \end{equation}
\end{itemize}

In the following, we also run simulations using $\mathcal{L}$ ({\sffamily Limit}) to choose the frequency. Our aim was not to show how efficient or how bad this technique is, but more to show that often, we observe rather counterintuitive results.

\subsection{Workloads and Simulation Architecture}
For the simulations we present bellow, we use two different sets of workloads. The first one is pretty simple, and quite theoretical. We use a set of 12 tasks, each of them having lengths uniformly distributed, between miscellaneous bounds, different from each other.
For the second set of simulations, we used several workloads coming from video decoding using H.264, which is used in our lab for some other experiments on a TI DaVinci DM6446 DVS processor. On Figure~\ref{fig:DivDistr}, we show the distribution of the 8 video clips we used, each with several thousands of frames.

We present here experimental results run for two different kinds of DVS processors (see for instance~\cite{Xu07} for details about characteristics): a XScale Intel processor (with frequencies 150, 400, 600, 800 and 1000MHz), and a PowerPC 405LP (with frequencies 33, 100, 266 and 333MHz). We took frequency change overhead into account, but the contribution of change overhead was usually negligible for all of the simulations we performed (lower that 0.1\% in most cases).
As a third CPU, we used the characteristics of XScale, but we disabled one of its available frequency (400MHz in the plots we show here), in order to highlight the advantage of using our approximation against round up approximation when the number of available frequencies is quite low.

\subsection{Simulations}
\begin{figure*}[ht]
\begin{center}
\newcommand{\lw}{0.23\linewidth}
\includegraphics[width=\lw]{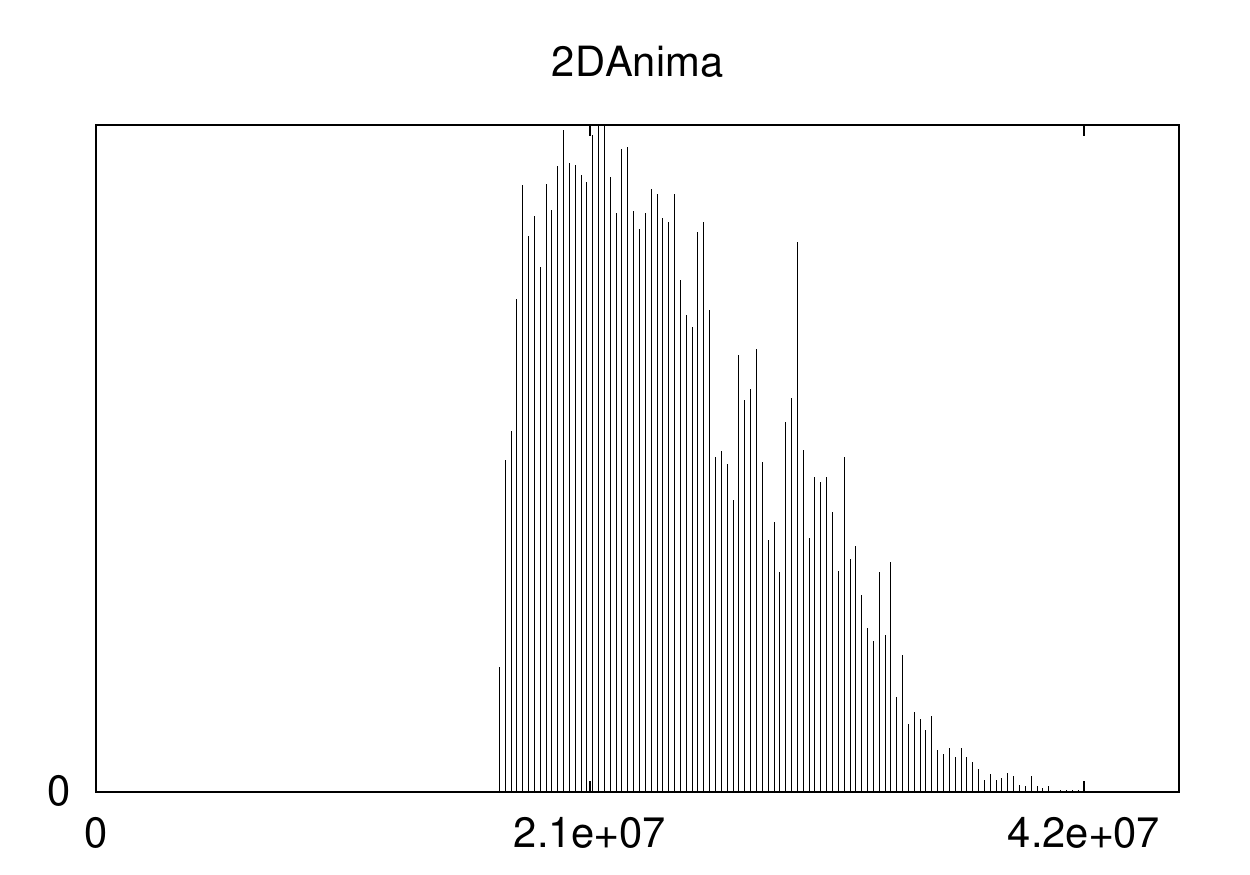}
\includegraphics[width=\lw]{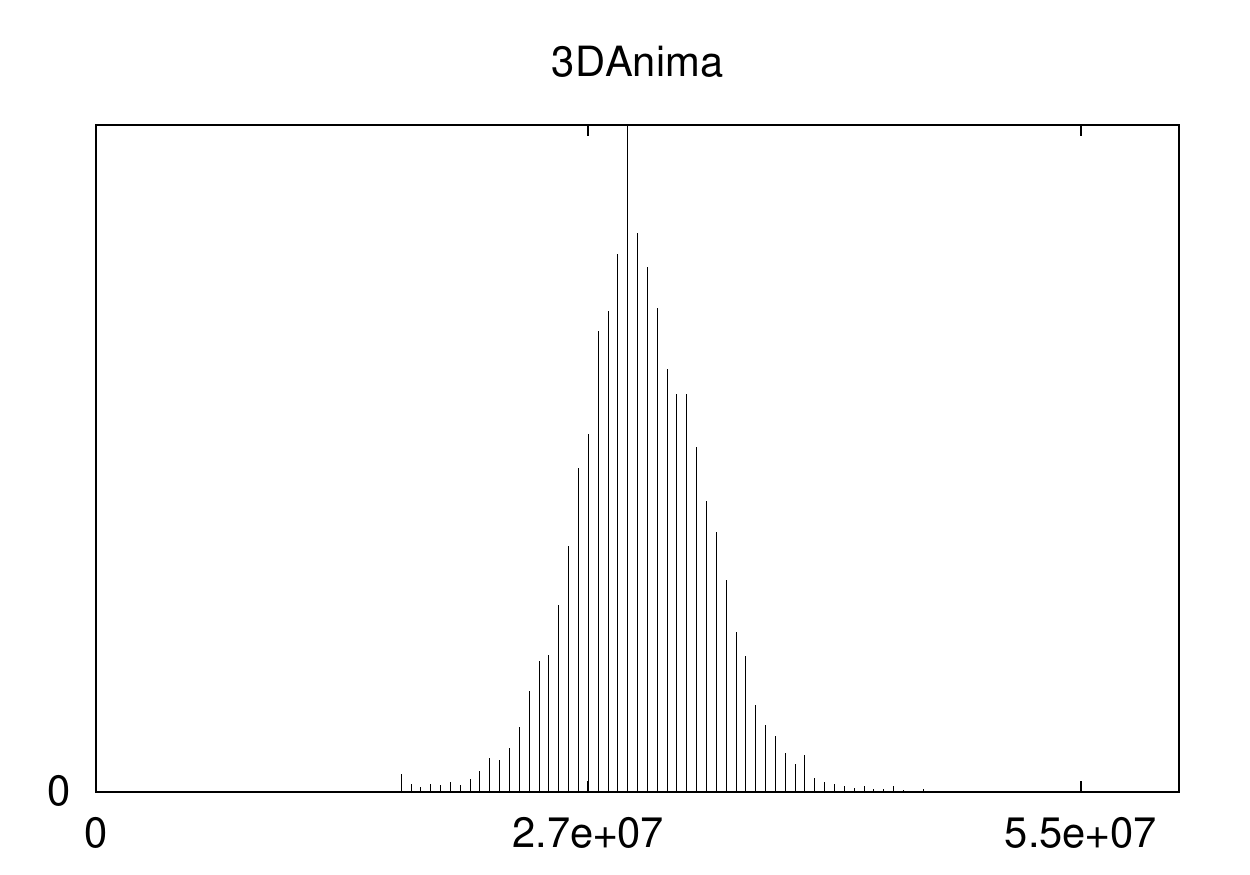}
\includegraphics[width=\lw]{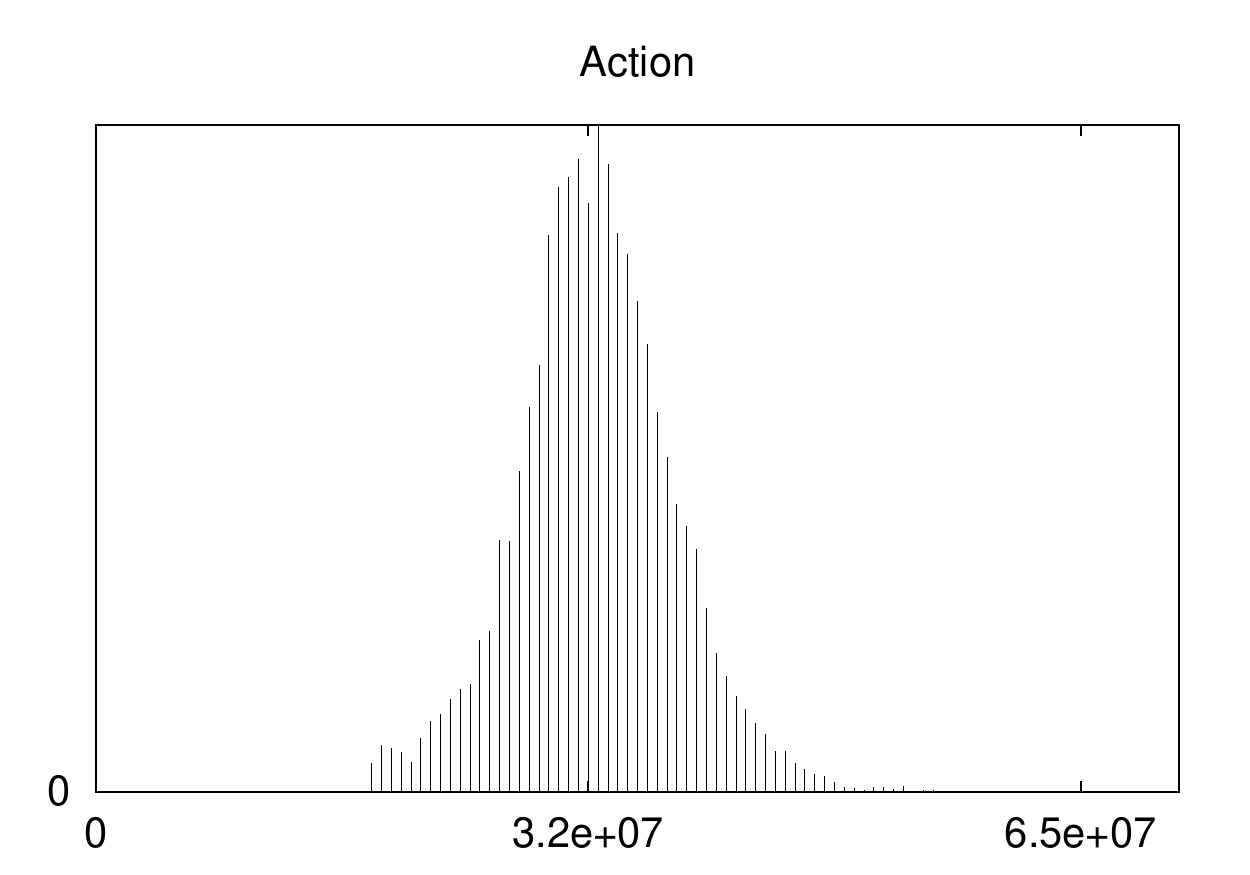}
\includegraphics[width=\lw]{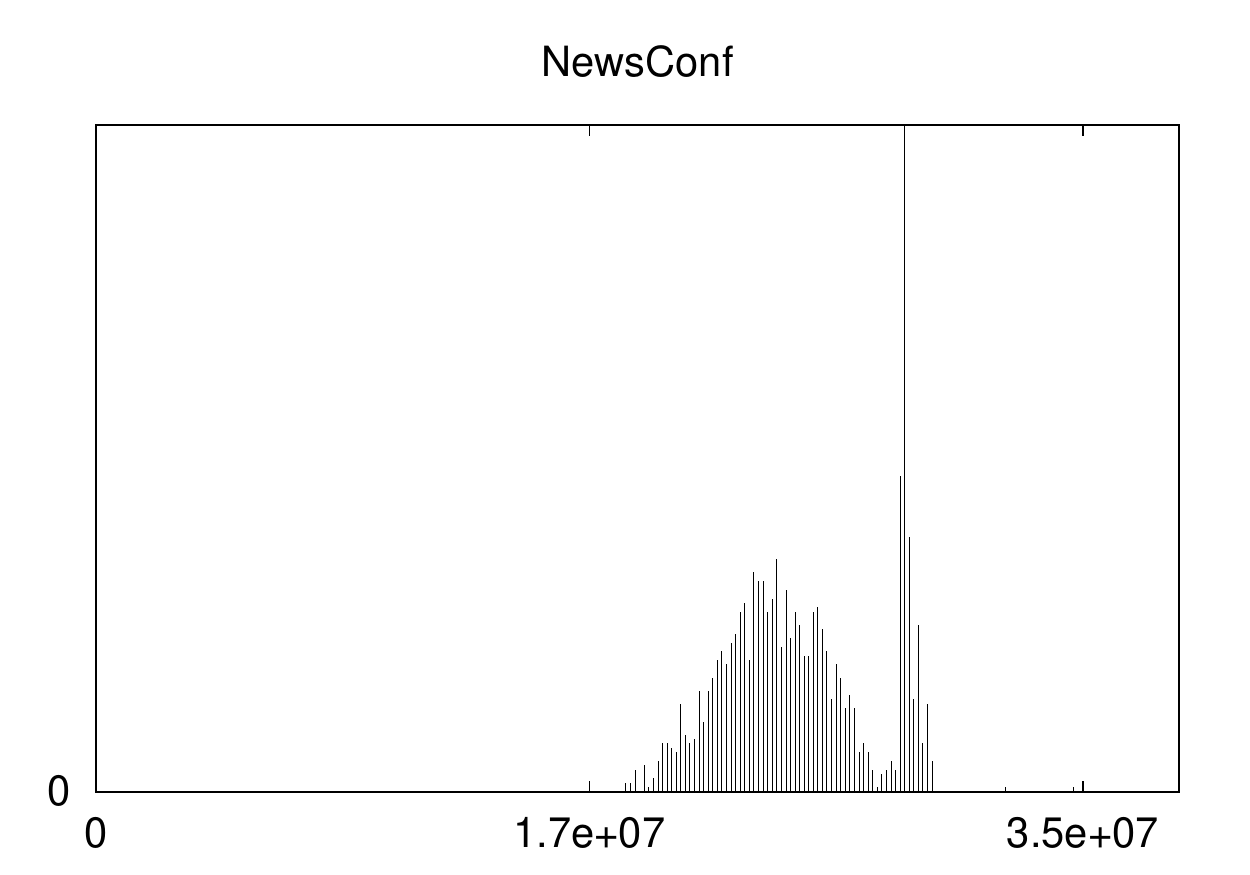}
\includegraphics[width=\lw]{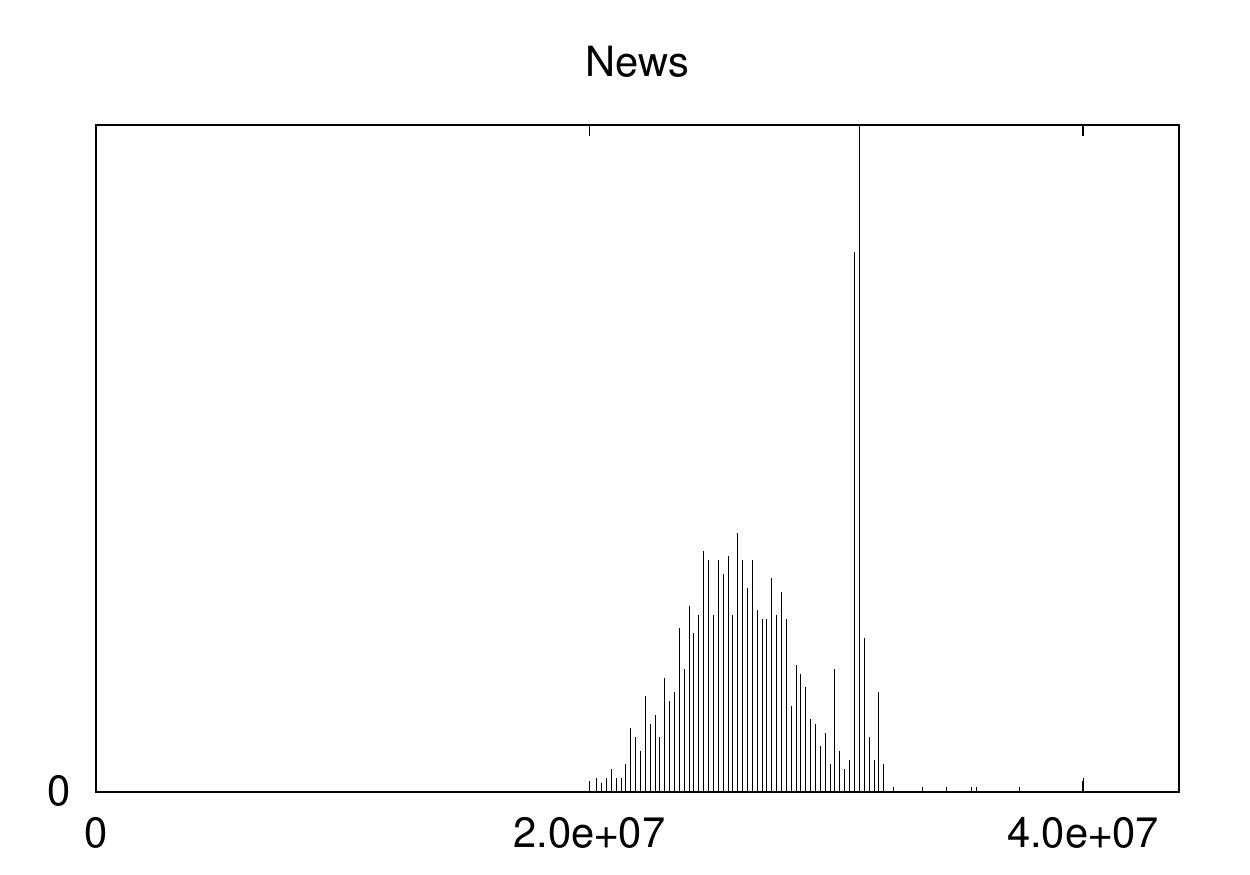}
\includegraphics[width=\lw]{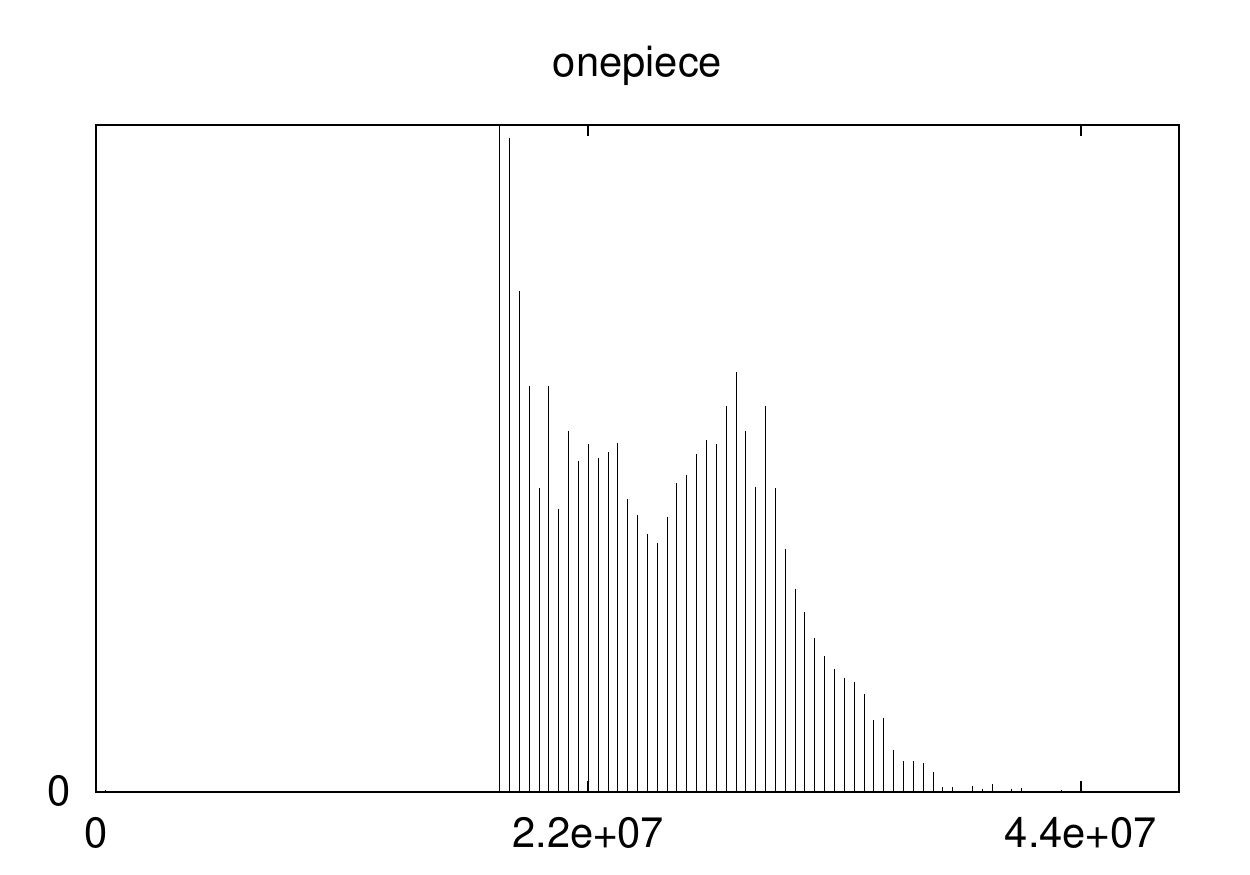}
\includegraphics[width=\lw]{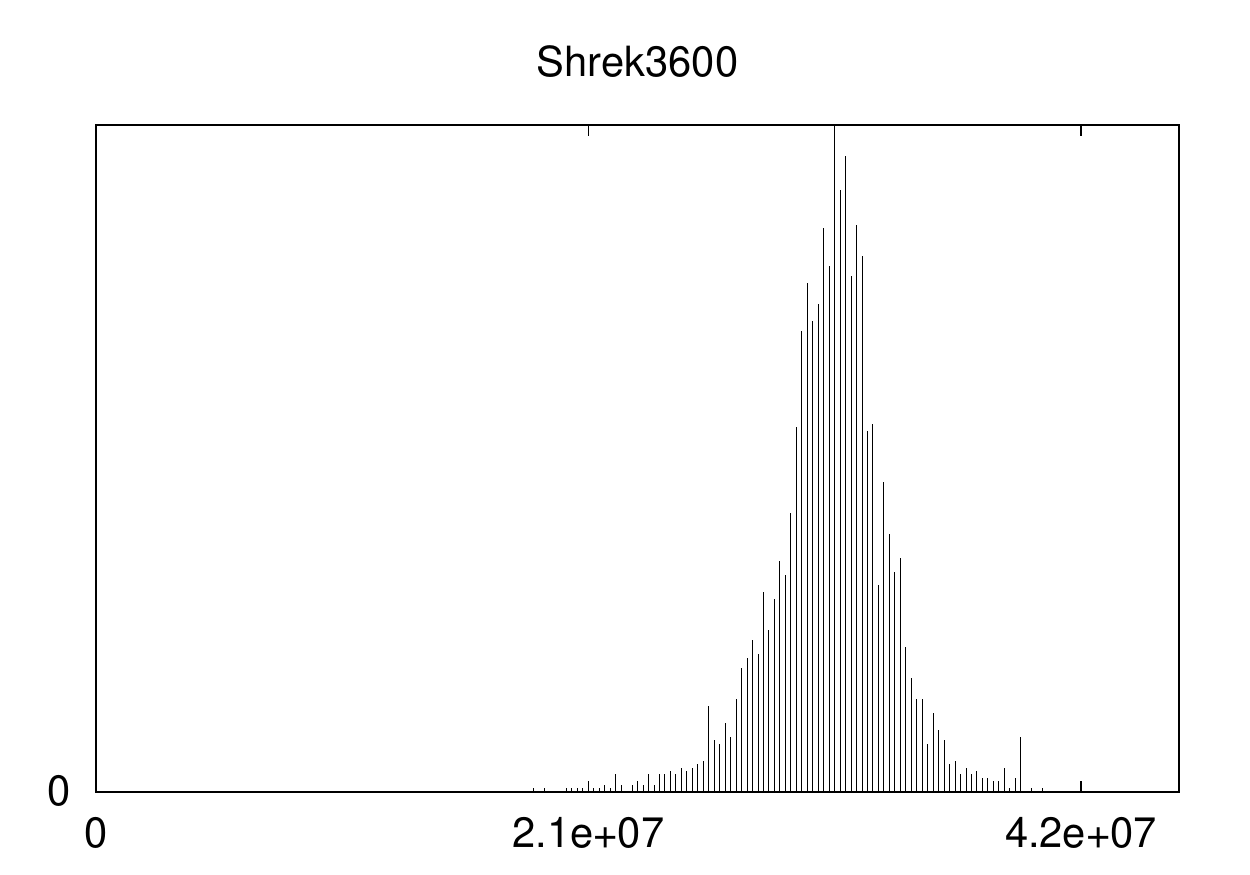}
\includegraphics[width=\lw]{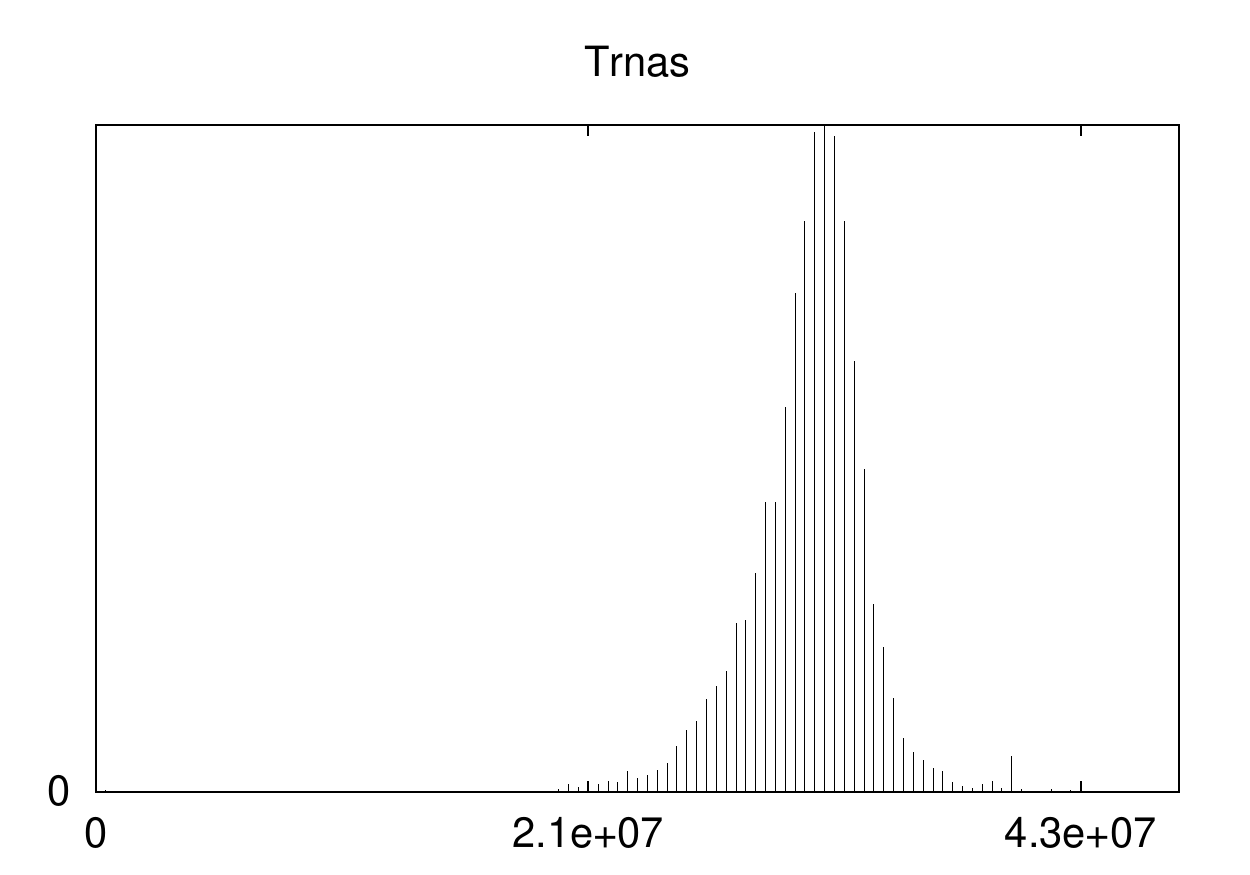}
\end{center}
\caption{\label{fig:DivDistr} Distribution of the number of cycles needed to decode different kinds of video, ranging from news streaming to complex 3D animations. The x-axis is the number of cycles, and the y-axis the probability.}
\end{figure*}

We performed a large number of simulations in order to compare the energy performance of ``round up'' and ``round to closest''. We compare several processor characteristics, and several job characteristics. We both use theoretical models and realistic values extracted from production systems.

For the figures we present here, we simulated the same system with different strategies computed with variations of Algorithm~\ref{alg:closest}, amongst {\sffamily DPM-S$^{\text{closest}}$} (Eq.~\eqref{eq:dpmsclosest}), {\sffamily DPM-S$^{\text{up}}$} (Eq.~\eqref{eq:dpmsup}),  {\sffamily PITDVS$^{\text{closest}}$} (Eq.~\eqref{eq:pitdvsclosest}), {\sffamily PITDVS$^{\text{up}}$} (Eq.~\eqref{eq:pitdvsup}) and {\sffamily Limit} (Algorithm~\ref{alg:limit}), computed the energy consumption, and presented the ratio of this energy to {\sffamily PITDVS$^{\text{closest}}$} or {\sffamily DPM-S$^{\text{closest}}$}. We then performed the same system, but for various deadlines, going from the deadline allowing to run any task at the lowest frequency ($D = \frac{1}{f_1} \sum_{i=1}^N w_i$), to the smallest deadline allowing to run any task at the higher frequency ($D = \frac{1}{f_M} \sum_{i=1}^N w_i$). We even used smaller deadlines, because this limit represents a frame where each task needs at the same time its WCEC, which has a very tiny probability to occur.
We can consider that decreasing the deadline boils down to increase the load: the smaller the deadline, the higher the average frequency. And quite intuitively, for small and large deadline (or frame length), we don't have any difference between strategies, because they all use always either the lowest (large deadline) or the highest (small deadline) frequency.

A first observation was that in many cases, the $S$-function of {\sffamily PITDVS$^{\text{up}}$} was already almost equal to {\sffamily Limit}. As a consequence, we could not observe any difference between {\sffamily PITDVS$^{\text{up}}$} and {\sffamily PITDVS$^{\text{closest}}$}. We can for instance see this on Figure~\ref{fig:Unif}, right plot: for deadlines between 0.1 and 0.06, we don't see any difference between {\sffamily PITDVS$^{\text{closest}}$} and {\sffamily Limit}.

In the first set of simulations (Figures~\ref{fig:Unif-DPMS} and \ref{fig:Unif}), we used 12 tasks, each of them having a uniformly distributed number of cycles, with miscellaneous parameters. On the PowerPC processor, we observe a large variety in performance comparison. According to the load (or the frame length), we see that {\sffamily PITDVS$^{\text{closest}}$} can gain around 30\% compared to {\sffamily PITDVS$^{\text{up}}$}, or lose almost 20\%, while we obtain similar comparison for {\sffamily DPM-S$^{\text{closest}}$} and {\sffamily DPM-S$^{\text{up}}$}, but with smaller values.

We observe also very abrupt and surprising variations, such as in Figure~\ref{fig:Unif}, middle and right, for {\sffamily Limit}, around 0.03. A closer look around to variations shows that they usually occurs when the frequency of $T_1$ changes. Indeed, as $T_1$ starts always at time 0, its speed does not really depends upon $S_1(t)$, but only upon $S_1(0)$. So when $D$ varies, $S_1(0)$ goes suddenly from one frequency to another one. Then a very slight variation of $D$ could have a big impact of each frame. Those slight variations do not have the same impact for other tasks, because of the stochastic nature of tasks length. For instance, if we slightly change $S_i$ ($i\ne 1$), it will only impact a few task speeds. But slight changes in $S_0$ have either no impact at all, or an impact on every task in every frame.

From those first figures, we can for sure not claim that doing a ``closest'' approach is always better than a ``upper bound''. But those simulations highlight that there are certainly situations where one approach is better than the other one, and situations with the other way around. System designers should then pay attention to the way they round continuous frequencies. With a very small additional effort, we can often do better than simply round up the original scheduling function.

For the second set of simulations (using real video workloads), on Figures~\ref{fig:Div-DPMS} and \ref{fig:Div}, we observe the same kind of differences as from the previous experiments: according to the configuration, one round method is better than the other one. With PowerPC configuration, {\sffamily PITDVS$^{\text{closest}}$} is better than {\sffamily PITDVS$^{\text{up}}$}, but {\sffamily DPM-S$^{\text{up}}$} seems to be better than {\sffamily DPM-S$^{\text{closest}}$}. However, with the XScale processor where we disabled one frequency, both ``closest'' methods are better than ``up'' methods. Remark that we observe the same kind of benefit by disabling another frequency than 400MHz.

From the many experiments we performed, it seems that our approach is especially interesting when the number of available frequencies is limited, which is not surprising. Indeed, the less available frequency, the further from the continuous model. As the two strategies we adapt where basically designed from continuous model, and as our adaptation attempts to be closer from the original strategy than the classical adaptation, we would have expected such behavior.

We have also observed than ``smooth'' systems such as the one with uniform distribution --- but we have simulated other distributions such as normal or bimodal normal distribution --- do not give smoother curves than with the realistic workload, even if several of them contain very chaotic data. The irregular behavior of our curve does not seem to be related to irregular data, but more to the fact that, as already mentioned slight variations in $S_0$ can have a big impact on the average energy.
In this paper, we do not present a huge number of simulations, because we do not claim that our approach is always better: what we present should be enough to persuade system designers to have a deeper look at the way they manage discretization.

\section{Conclusions and Future Work}
\label{sec:ccl}
The aim of our work was twofold. First, we presented a simple schedulability condition for frame-based low-power stochastic real-time systems. Thanks to this condition, we are able to quickly check that any scheduling function guarantees the schedulability of the system, even when frequency change overheads are taken into account. This test can either be used off-line to check that a scheduling function is schedulable, or on-line, after some parameter changes, to check whether the functions can still be used.

The second contribution of this paper was to use this schedulability condition in order improve the way a strategy developed for systems with continuous speeds can be adapted for systems with a discrete set of available speeds. We show that our approach is not always better that the classical one consisting in rounding up to the first available frequency, but can in some circumstances, give a gain up to almost 40\% in the simulations we presented.

Our future work includes several aspects. First, by running much more simulations, we would like to identify more precisely when our approach is better than the classical one. It would allow system designers to be able to choose the approach to use without running simulation, or making experiments on their system.

Another aspect we would like to consider is to have a deeper look to how the schedulability test we provide will allow to improve the robustness of a system. If particular, if we observe that a job has required more than its (expected) worst case number of cycles, how can we adapt temporarily our system in order to improve its schedulability, before we can compute the new set of functions, using those new parameters.

\bibliographystyle{acm}
\nocite{*}

\bibliography{SchedFrame}
\end{document}